\documentclass[acmsmall, authorversion=true]{acmart}

\usepackage{prelude}

\newtoggle{widecol}
\toggletrue{widecol}

\newcommand{\captiondetail}{%
  \justifying\small\mbox{}\\[-0.7\baselineskip]\noindent}

\newcommand{\appsub}[3]{\appopt{\subopt{#1}{#2}}{#3}}
\newcommand{\sqappsub}[3]{\sqappopt{\subopt{#1}{#2}}{#3}}

\newcommand{\worst}{\mathrm{worst}}
\newcommand{\new}{\mathrm{new}}
\newcommand{\old}{\mathrm{old}}
\newcommand{\res}{\mathrm{res}}
\newcommand{\wait}{\mathrm{wait}}
\newcommand{\rJw}{r_J}
\newcommand{\aJw}{a_J''}
\newcommand{\rworst}[2][]{\appsub{r^\worst}{#1}{#2}}
\newcommand{\aworst}[2][]{\appsub{a^\worst}{#1}{#2}}
\newcommand{\Xnew}[2][]{\sqappsub{X^\new}{#1}{#2}}
\newcommand{\Xold}[2][]{\sqappsub{X^\old}{#1}{#2}}
\newcommand{\Yold}[2][]{\sqappsub{Y^\old}{#1}{#2}}
\newcommand{\Tres}{T^\res}
\newcommand{\Twait}{T^\wait}
\newcommand{\rhonew}[2][]{\sqappsub{\rho^\new}{#1}{#2}}
\newcommand{\rhoold}[2][]{\sqappsub{\rho^\old}{#1}{#2}}
\newcommand{\Bnew}[2][]{\sqappsub{B^\new}{#1}{#2}}
\newcommand{\tlXnew}[3][]{\appopt{\sqappsub{\tl{X}^\new}{#1}{#2}}{#3}}
\newcommand{\tlXold}[3][]{\appopt{\sqappsub{\tl{X}^\old}{#1}{#2}}{#3}}
\newcommand{\tlYold}[3][]{\appopt{\sqappsub{\tl{Y}^\old}{#1}{#2}}{#3}}
\newcommand{\tlBnew}[3][]{\appopt{\sqappsub{\tl{B}^\new}{#1}{#2}}{#3}}
\newcommand{\tlTres}{\tl{T}^\res}
\newcommand{\tlTwait}{\tl{T}^\wait}

\newcommand{\2}{\mathbf{CoinFlip}}
\newcommand{\dnil}{\emptyset}

\newcommand{\olfloor}[1]{%
  \mkern 4mu\overline{\mkern-4mu\floor{#1}\mkern-2.9mu}\mkern 2.9mu}
\newcommand{\olfloorplus}[3]{%
  \mkern 4mu\overline{\mkern-4mu\floor{#1} + #2\mkern-#3}\mkern #3}
\newcommand{\olceil}[1]{%
  \mkern 4mu\overline{\mkern-4mu\ceil{#1}\mkern-2.9mu}\mkern 2.9mu}

\usetikzlibrary{intersections, backgrounds}
\tikzset{%
  rank plot/.style={x=14, y=4, thick, scale=\iftoggle{widecol}{1.14}{1}},
  axis/.style={thick},
  primary/.style={ultra thick, color=cyan},
  secondary/.style={ultra thick, color=green!41!yellow!78!black},
  cutoff/.style={line width=3.2pt, dotted, color={red!29!yellow!91!black}},
  rworst/.style={line width=3.2pt, dotted, color=magenta},
  rworst point/.style={ultra thick, color=magenta, fill=white},
  aworst/.style={ultra thick, color=green!41!yellow!78!black},
  discarded/.style={color=red!10},
  original/.style={color=green!84!yellow!14},
  recycled/.style={color=green!84!yellow!14},
  guide/.style={thick, densely dotted}}

\makeatletter
\catcode`\$=3 

\def\@above{above}
\def\@below{below}
\def\@negative{-}
\newcommand{\downmode}{%
  \def\@above{below}
  \def\@below{above}
  \def\@negative{}}

\newcommand{\projx}[1]{($ (0,0)!#1!(1,0) $)}
\newcommand{\projxheight}[2]{($ (0,#2)!#1!(1,#2) $)}
\newcommand{\midpt}[2]{($ #1!0.5!#2 $)}

\newcommand{\xguide}[3][]{%
  \draw[axis] (#2, 0) -- (#2, \@negative0.7)
  node[\@below] {\ifempty{#1}{$#2$}{#1}\vphantom{by}};
  \draw[guide] (#2, 0) -- (#2, #3);}
\newcommand{\yguide}[3][]{%
  \draw[axis] (0, #3) -- (-0.2, #3)
  node[left] {\ifempty{#1}{$#3$}{#1}};
  \draw[guide] (0, #3) -- (#2, #3);}
\newcommand{\yguidenolabel}[2]{%
  \draw[guide] (0, #2) -- (#1, #2);}
\newcommand{\xguidept}[2]{%
  \ifempty{#1}{}{%
    \draw[axis] \projx{#2} -- \projxheight{#2}{\@negative0.7}
    node[\@below] {#1\vphantom{by}}};
  \draw[guide] \projx{#2} -- #2;}

\newcommand{\axes}[4]{%
  \draw[axis, ->] (-0.2, 0)
  node[left] {$0$} -- ({#1},0) node[right] {#3};
  \draw[axis, ->] (0, \@negative0.7)
  node[\@below] {$0\vphantom{by}$} -- (0,#2) node[\@above] {#4};}

\newcommand{\point}[2][primary]{%
  \filldraw[#1, very thick] #2 circle (1.8mu);}

\newcommand{\jump}[4][primary]{%
  \point[#1, fill=white]{(#2, #3)}
  \point[#1]{(#2, #4)}}

\newcommand{\snakefirst}[2]{#1 parabola bend #1 \midpt{#1}{#2}}
\newcommand{\snakesecond}[2]{\midpt{#1}{#2} parabola bend #2 #2}
\newcommand{\snake}[2]{\snakefirst{#1}{#2} -- \snakesecond{#1}{#2}}

\newcommand{\interval}[4]{%
  \filldraw[#1] \projx{#2} rectangle #3;
  \begin{scope}
    \node (A) at \projx{#2} {};
    \node (B) at \projx{#3} {};
    \node[\@above] at \midpt{(A)}{(B)} {\footnotesize #4\vphantom{by}};
  \end{scope}}
\newcommand{\discarded}[2]{%
  \interval{discarded}{#1}{#2}{discarded}}
\newcommand{\original}[2]{%
  \interval{original}{#1}{#2}{original}}
\newcommand{\recycled}[3][]{%
  \interval{recycled}{#2}{#3}{\ifempty{#1}{}{$#1$-}recycled}}

\newcommand{\Isnake}{%
  \draw[primary, name path=rank]
  \snake{(0, 3)}{(1.5, 7)}
  -- \snake{(1.5, 7)}{(2, 4)}
  -- \snake{(2, 4)}{(3.25, 12)}
  -- \snake{(3.25, 12)}{(5.25, 6)}
  -- \snake{(5.25, 6)}{(5.75, 5)}
  -- \snake{(5.75, 5)}{(8.25, 9)}
  -- \snake{(8.25, 9)}{(10.5, 3)}
  -- \snake{(10.5, 3)}{(12.5, 9)};
  \draw[primary, dotted] \snakefirst{(12.5, 9)}{(14.9, 7)};}

\newcommand{\Jsnake}{%
  \draw[primary, name path=rank]
  \snake{(0, 6)}{(1.5, 3)}
  -- \snake{(1.5, 3)}{(2.5, 5)}
  -- \snake{(2.5, 5)}{(5.5, 12)}
  -- \snake{(5.5, 12)}{(7, 4)}
  -- \snake{(7, 4)}{(8, 8)}
  -- \snake{(8, 8)}{(9.5, 2)}
  -- \snake{(9.5, 2)}{(11, 8)}
  -- \snake{(11, 8)}{(12.5, 5)};
  \draw[primary, dotted] \snakefirst{(12.5, 5)}{(14.9, 7)};}

\newcommand{\Ksnake}{%
  \draw[primary, name path=rank]
  \snake{(0, 2)}{(2, 7)}
  -- \snake{(2, 7)}{(4, 3)}
  -- \snake{(4, 3)}{(6, 10)}
  -- \snake{(6, 10)}{(7.5, 4)}
  -- \snake{(7.5, 4)}{(9.5, 12)}
  -- \snake{(9.5, 12)}{(11.5, 4)}
  -- \snake{(11.5, 4)}{(12.5, 5)};
  \draw[primary, dotted] \snakefirst{(12.5, 5)}{(14.9, 3)};}

\catcode`\$=13 
\makeatother

\begin{document}

\setcopyright{acmcopyright}
\acmJournal{POMACS}
\acmYear{2018}
\acmVolume{2}
\acmNumber{1}
\acmArticle{16}
\acmMonth{3}
\acmPrice{15.00}
\acmDOI{10.1145/3179419}



\title[SOAP: One Clean Analysis of All Age-Based Scheduling Policies]{%
  SOAP: One Clean Analysis of All
  \iftoggle{widecol}{\\}{}%
  Age-Based Scheduling Policies}

\author{Ziv Scully}
\affiliation{%
  \institution{Carnegie Mellon University}
  \department{Computer Science Department}
  \streetaddress{5000 Forbes Ave}
  \city{Pittsburgh}
  \state{PA}
  \postcode{15213}
  \country{USA}}
\email{zscully@cs.cmu.edu}

\author{Mor Harchol-Balter}
\affiliation{%
  \institution{Carnegie Mellon University}
  \department{Computer Science Department}
  \streetaddress{5000 Forbes Ave}
  \city{Pittsburgh}
  \state{PA}
  \postcode{15213}
  \country{USA}}
\email{harchol@cs.cmu.edu}

\author{Alan Scheller-Wolf}
\affiliation{%
  \institution{Carnegie Mellon University}
  \department{Tepper School of Business}
  \streetaddress{5000 Forbes Ave}
  \city{Pittsburgh}
  \state{PA}
  \postcode{15213}
  \country{USA}}
\email{awolf@andrew.cmu.edu}

\begin{abstract}
  We consider an extremely broad class of M/G/1 scheduling policies
  called SOAP: Schedule Ordered by Age-based Priority.
  The SOAP policies include almost all scheduling policies in the literature
  as well as an infinite number of variants
  which have never been analyzed, or maybe not even conceived.
  SOAP policies range from classic policies, like
  first-come, first-serve~(FCFS), foreground-background~(FB),
  class-based priority,
  and shortest remaining processing time~(SRPT); to much more complicated scheduling rules,
  such as the famously complex Gittins index policy
  and other policies in which a job's priority
  changes arbitrarily with its age.
  While the response time of policies in the former category is well understood,
  policies in the latter category have resisted response time analysis.
  We present a universal analysis of all SOAP policies,
  deriving the mean and Laplace-Stieltjes transform of response time.
\end{abstract}

\maketitle

\section{Introduction}

Analyzing the response time of scheduling policies in the M/G/1 setting
has been the focus of thousands of papers over the past half century,
from classic early works
\citep{book_kleinrock, fb_analysis_schrage, m/g/1_kendall,
  nonpreemptive_takacs, ps_analysis_kleinrock, srpt_analysis_schrage,
  vacations_fuhrmann}
to more recent works in the SIGMETRICS community
\citep{fb_optimality_misra, fairness_wierman, mlps_delay_aalto,
  multiclass_ayesta, ps_asymptotics_borst, ps_heavy_robert, ps_beyond_aalto,
  smart_insensitive_wierman, smart_preventing_wierman, smart_epsilon_wierman,
  srpt_heavy_zwart, vacations_ayesta}.
Examples of common scheduling policies include
\begin{itemize}
\item
  \emph{first-come, first-served}~(FCFS),
  which serves jobs nonpreemptively in the order they arrive;
\item
  \emph{class-based priority},
  which serves the job of highest priority class,
  possibly preemptively and possibly nonpreemptively;
\item
  \emph{shortest remaining processing time}~(SRPT),
  which preemptively serves the job with the least remaining time;
\item
  \emph{foreground-background}~(FB),
  which preemptively serves the job that has received the least service so far;
  and
\item
  \emph{processor sharing}~(PS),
  which concurrently serves all jobs in the system
  at the same rate.
\end{itemize}
In just these few examples
we see a variety of features represented:
preemptible jobs, nonpreemptible jobs,
prioritizing by class, prioritizing by job size,
and prioritizing by service received so far, or \emph{age}.
Each policy requires a custom response time analysis
that takes into account its particular combination of features.

Although there has been much success in analyzing the response time of
specific scheduling policies in the M/G/1 setting,
such as those listed above,
results are ad-hoc and \emph{limited to relatively simple policies}.
Analyzing variants of the above simple policies,
let alone fundamentally different policies,
is an open problem.
For instance, none of the following scenarios have been analyzed before.
\begin{itemize}
\item
  Suppose we have exact size information for some ``sized'' jobs
  but not other ``unsized'' jobs.
  We run SRPT on sized jobs and FB on unsized jobs,
  meaning that we serve the sized job of minimum remaining time
  or unsized job of minimum age,
  whichever measurement is smaller.
\item
  Suppose we have jobs that are neither fully preemptible
  nor fully nonpreemptible
  but instead preemptible only at specific ``checkpoint'' ages.
  We run a preemptive policy, such as SRPT or FB,
  but only preempt jobs when they reach checkpoint ages.
\item
  The \emph{Gittins index policy}
  \citep{book_gittins, m/g/1_gittins_aalto},
  long known to be optimal for minimizing mean response time
  in the M/G/1 queue\footnote{%
    While SRPT is optimal when exact job sizes are known,
    the Gittins index policy, of which SRPT is a special case,
    is optimal even when only size distributions are known.},
  has only been analyzed in certain special cases
  \citep{multiclass_ayesta, rs_slowdown_hyytia}.
  In general, the Gittins index policy can have
  a complex priority scheme \citep{mlps_gittins_aalto}
  which, while known to perform optimally,
  has not been analyzed before in its general form.
\end{itemize}

Approaching the above examples with state-of-the-art techniques,
if possible at all,
would require an ad-hoc analysis for each scenario.
We seek \emph{general principles and techniques} for response time analysis
that apply to not just the above examples
but to as many scheduling policies as possible,
even those not yet imagined.

\subsection{Contributions}

We introduce \emph{SOAP}, a \emph{universal framework}
for defining and analyzing M/G/1 scheduling policies.
The SOAP framework can analyze \emph{any SOAP scheduling policy},
which includes nearly any policy where a job's priority
depends on its own characteristics: class, size, age, and so on.
Specifically, we make the following contributions.
\begin{itemize}
\item
  We \emph{define} the class of SOAP policies (\fref{sec:model}),
  a broad class of policies that includes
  the three unsolved examples above as well as many other policies,
  from practical scenarios to policies not yet imagined.
  We \emph{encode} many policies old and new
  as SOAP policies (\fref{sec:examples}).
\item
  We give a \emph{universal response time analysis}
  that works for any SOAP policy (\fref{sec:response_time}),
  obtaining closed forms for
  the mean (\fref{thm:response_time_mean})
  and Laplace-Stieltjes transform (\fref{thm:response_time_transform}).
  In particular, we apply our results to
  \emph{previously intractable analyses} (\fref{sec:applications}),
  such as the response time of the Gittins index policy.
\end{itemize}

In defining and analyzing SOAP policies,
there are two major technical challenges.
The first major challenge is that
to have a single analysis apply to many scheduling policies at once,
we need to \emph{express all such policies within a single framework}.
The SOAP framework encodes a scheduling policy
as a \emph{rank function},
which maps each job to a priority level, or \emph{rank}.
All SOAP policies are based on a single rule:
always serve the job of \emph{minimal rank}.
For example, in a preemptive class-based priority system,
a job's rank is its class (\fref{ex:soap_classes}),
whereas in SRPT, a job's rank is its remaining time (\fref{ex:soap_srpt}).
Rank functions can express a huge variety of policies,
from virtually all classic policies (\fref{sub:examples_previous})
to complex policies
which have never been analyzed before (\fref{sub:examples_new}).
A notable exception is PS, which does not fit into the SOAP framework.

The second major challenge is to
\emph{analyze policies with arbitrary rank functions}.
In particular, nearly all previously analyzed scheduling policies,
when expressed as SOAP policies,
have rank functions that are \emph{monotonic} in age.
For example, under SRPT, a job's rank decreases with age,
making it less and less likely to be preempted by another job,
while under FB, a job's rank increases with age,
making it more and more likely to be preempted by another job.
Unfortunately, the techniques used in the past
to analyze policies with monotonic rank functions
\emph{break down for arbitrary nonmonotonic rank functions},
which appear, for instance,
when studying the Gittins index policy (\fref{ex:soap_gittins})
and jobs that are preemptible only at certain checkpoints
(\fref{ex:soap_checkpoints}).
We develop \emph{new analytical tools}
that work for arbitrary rank functions (\fref{sec:key_ideas}).

\subsection{Related Work}

Our work on SOAP policies follows in the tradition of
analyses that address an entire class of policies at once.
Two such classes are \emph{SMART} \citep{smart_insensitive_wierman}
and \emph{multilevel processor sharing}~(MLPS) \citep{book_kleinrock}.
\begin{itemize}
\item
  The SMART class includes all policies that satisfy certain criteria
  that ensure they prioritize small jobs over large ones,
  such as SRPT and PSJF (\fref{ex:soap_srpt}).
  Some recent work on SMART policies includes
  analyzing the tail behavior of response time \citep{smart_preventing_wierman}
  and characterizing the tradeoff between
  accuracy of size estimates and response time \citep{smart_epsilon_wierman}.
\item
  The MLPS class consists of policies that
  divide all jobs in the system into echelons based on age,
  then serves jobs in the youngest echelon according to FCFS, FB, or PS.
  Some recent work on MLPS policies includes
  optimally choosing the age echelon cutoffs \citep{mlps_threshold_osipova}
  and connecting MLPS to the Gittins index policy \citep{mlps_gittins_aalto}.
\end{itemize}
While the SMART and MLPS classes have nearly no overlap,
the SOAP class includes many policies from \emph{both} classes.
Specifically, the SMART* subclass of SMART \citep{smart_insensitive_wierman}
and MLPS policies which do not use PS are all SOAP policies.

A particularly important SOAP policy is the Gittins index policy
\citep{book_gittins, m/g/1_gittins_aalto},
which minimizes mean response time in the M/G/1 queue
when job sizes are not known.
The Gittins index policy has a rather complex definition,
but recent work \citep{m/g/1_gittins_aalto, mlps_gittins_aalto}
has revealed some of its structural properties.
\Citet{multiclass_ayesta} analyze a specific case of the Gittins index policy
for a multiclass M/G/1 queue where each class's job size distribution
has the \emph{decreasing hazard rate}~(DHR) property.
Using the SOAP framework, we can analyze the Gittins index policy
for \emph{arbitrary size distributions}.
\Citet{rs_slowdown_hyytia} show that for jobs with known sizes,
a weighted version of the Gittins index yields the
\emph{shortest processing time product}~(SPTP) policy
and that this policy minimizes mean \emph{slowdown},
which is the ratio of a job's response time to its size.
SPTP is a SOAP policy, so the SOAP framework can obtain
the Laplace-Stieltjes transform of slowdown for SPTP,
extending the previous mean analysis.

\section{System Model and SOAP Policies}
\label{sec:model}

We consider work-conserving scheduling policies for the M/G/1 queue.
We write $\lambda$ for the total arrival rate
and $X$ for the overall job size distribution.
We assume a stable system, meaning $\lambda\E{X} < 1$,
and a preempt-resume model,
meaning preemption and processor sharing are permitted
without penalty or loss of work.

\subsection{Descriptors}
\label{sub:descriptors}

Scheduling algorithms use information about jobs in the system
when deciding which job to serve.
We can divide this information into two types:
\emph{static} and \emph{dynamic}.
\begin{itemize}
\item
  \emph{Static} information about a job
  is revealed when it enters the system and never changes.
  For example, in a system with multiple job classes,
  a job's class would be static information,
  and in a system where exact job sizes are known,
  a job's exact size would be static information.
  We call a job's static information its \emph{descriptor}
  and write~$\mc{D}$ for the set of descriptors.
  A job's descriptor~$d$ determines its size distribution
  $X_d = (X \mid \text{job has descriptor } d)$.
\item
  \emph{Dynamic} information about a job changes as a job is served.
  In this paper,
  the only dynamic information about a job is its \emph{age},
  the amount of time it has been served.
  The set of possible ages is~$\R_{\geq 0}$.
\end{itemize}

Descriptors are often tuples.
To distinguish descriptors from ranks,
another type of tuple introduced in \fref{sub:soap_policies},
we write descriptors in [square brackets]
and ranks in $\langle$angle brackets$\rangle$.

\begin{example}
  \label{ex:descriptors}
  Consider a system with a set of job classes~$\mc{K}$,
  where $X_k$ is the size distribution of class $k \in \mc{K}$.
  Depending on what information is known to the scheduler,
  the set of descriptors~$\mc{D}$ may be one of several options.
  \begin{itemize}
  \item
    If jobs do not reveal their exact size upon entering the system,
    then $\mc{D} = \mc{K}$,
    because the only static information we have about each job is its class.
    The size distribution of jobs with descriptor~$k$ is simply~$X_k$.
  \item
    If jobs reveal their exact size upon entering the system,
    then $\mc{D} = \mc{K} \times \R_{\geq 0}$,
    because we know each job's class and size.
    The size distribution of jobs with descriptor $[k, x]$
    is $X_{[k, x]} = x$,
    the deterministic distribution with value~$x$.
  \item
    If only some jobs reveal their exact size, then
    \begin{equation*}
      \mc{D} = \mc{K} \times (\R_{\geq 0} \cup \{?\}),
    \end{equation*}
    because some jobs have known exact size $x \in \R_{\geq 0}$
    while others have unknown size, which we denote by~$?$.
    The size distributions are
    $X_{[k, x]} = x$ for $x \in \R_{\geq 0}$
    and $X_{[k, ?]} = X_k$.
  \end{itemize}
\end{example}

We require that the descriptors of jobs
must be chosen i.i.d. according to a fixed distribution.
For instance, in \fref{ex:descriptors},
each job's class must be chosen i.i.d.,
and in the third scenario,
having each job independently reveal its size
with probability~$1/2$ is permitted,
but having alternating arrivals reveal their sizes is not.

\subsection{SOAP Policies and Rank Functions}
\label{sub:soap_policies}

A \emph{SOAP scheduling policy} is a preemptive priority policy
where a job's descriptor and age determine its priority.
SOAP is an acronym for \emph{Schedule Ordered by Age-based Priority}.
Specifically, a SOAP policy is specified by the following ingredients:
\begin{itemize}
\item
  a set~$\mc{R}$ of \emph{ranks},
\item
  a strict total order $\prec$ on $\mc{R}$, and
\item
  a \emph{rank function} assigning a rank $r(d, a)$
  to each pair of descriptor~$d$ and age~$a$,
  \begin{equation*}
    r : \mc{D} \times \R_{\geq 0} \to \mc{R}.
  \end{equation*}
\end{itemize}
The defining property of SOAP policies is the following.
\begin{quote}
  \textbf{Every moment in time,
    a SOAP policy serves the job of \emph{minimum rank}.}
\end{quote}
Ties between jobs of the same rank
are broken using a tiebreaking rule.
For simplicity of exposition,
we focus on \emph{first-come, first-served}~(FCFS) tiebreaking,
which, among the jobs of minimal rank,
serves the job that arrived to the system first.
Our results also apply to SOAP policies that use
\emph{last-come, first-served}~(LCFS) tiebreaking,
which we defer to \fref{app:lcfs}\footnote{%
  Sometimes ties for minimum rank lead to processor sharing,
  and these ties do not require tiebreaking.
  See \fref{ex:soap_fb} and \fref{app:rank_function} for details.}.

Suppose job~$J$ has descriptor~$d_J$ and age~$a_J$
and job~$K$ has descriptor~$d_K$ and age~$a_K$.
We say $J$ \emph{outranks} $K$
if $r(d_J, a_J) \prec r(d_K, a_K)$
or both $r(d_J, a_J) = r(d_K, a_K)$ and $J$ arrived before~$K$.
SOAP policies always serve the job that outranks all other jobs in the system.

A great many SOAP policies can be expressed using $\mc{R} = \R$,
in which case $\prec$ is the usual ordering~$<$ on $\R$.
However, it is often convenient to have a nested rank structure
in which jobs are first prioritized by a \emph{primary rank},
then by a \emph{secondary rank},
and only then by the tiebreaking rule.
We express such nested ranks using $\mc{R} = \R^2$.
Each rank is a pair $\angle{r_1, r_2}$
of primary rank~$r_1$ and secondary rank~$r_2$,
and $\prec$ is the lexicographic ordering:
$\angle{r_1, r_2} \prec \angle{r_1', r_2'}$
if $r_1 < r_1'$ or both $r_1 = r_1'$ and $r_2 < r_2'$.
We write primary and secondary ranks of a state $(d, a)$
as $r_1(d, a)$ and $r_2(d, a)$, respectively, so
\begin{equation*}
  r(d, a) = \angle{r_1(d, a), r_2(d, a)}.
\end{equation*}

When specifying a SOAP policy,
we usually leave the choice of $\mc{R}$ unstated,
as it is implied from the formula for the rank function.
Our results apply to $\mc{R} = \R^n$ ordered lexicographically
for any $n \geq 1$,
and they easily generalize to other choices of~$\mc{R}$.

We will devote much time to discussing how jobs' ranks change with age.
Thus, when we call a rank function ``monotonic'' or similar,
we mean that it is so with respect to age, not descriptor.

For a SOAP policy to be well-defined,
its rank function must satisfy some technical conditions,
which are given in \fref{app:rank_function}.

\section{SOAP Policies Are Everywhere}
\label{sec:examples}

\subsection{Previously Analyzed SOAP Policies}
\label{sub:examples_previous}

\begin{example}
  \label{ex:soap_fb}
  The \emph{foreground-background}~(FB) policy is a SOAP policy.
  It uses no static information, so $\mc{D} = \{\dnil\}$,
  where $\dnil$ is a ``placeholder'' descriptor assigned to every job.
  FB always serves the job of least age,
  so it has rank function $r(\dnil, a) = a$.
  It is likely that many jobs are tied for minimum rank under FB,
  but whichever job is served immediately loses minimum status,
  resulting in a processor-sharing effect.
\end{example}

There are always many rank functions that encode the same SOAP policy.
For instance, any rank function monotonically increasing in age,
such as $r(\dnil, a) = a^2$, also describes FB.

\begin{example}
  \label{ex:soap_fcfs}
  The \emph{first-come, first-served}~(FCFS) policy is a SOAP policy.
  It uses no static information, so $\mc{D} = \{\dnil\}$.
  FCFS is nonpreemptive,
  which is equivalent to always serving the job of maximal age,
  so it has rank function $r(\dnil, a) = -a$.
  FCFS tiebreaking plays a crucial role by breaking ties between
  jobs of age~$0$.
\end{example}

Once again, there are multiple rank functions that describe FCFS.
In particular, a constant rank function yields FCFS
due to the tiebreaking rule,
but we prefer the given encoding because it makes it clear that
FCFS is a \emph{nonpreemptive} policy.
As the following examples demonstrate,
using primary rank~$-a$ is a general way to indicate nonpreemptiveness
in a rank function.

\begin{example}
  \label{ex:soap_classes}
  Consider a system with classes $\mc{K} = \{1, \dots, n\}$
  where jobs within each class are served in FCFS order
  but class~$1$ has highest priority,
  class~$2$ has next-highest priority, and so on.
  The \emph{nonpreemptive priority}
  and \emph{preemptive priority} policies are SOAP policies.
  Both policies use job class as static information, so $\mc{D} = \mc{K}$.
  \begin{itemize}
  \item
    Nonpreemptive priority has rank function
    $r(k, a) = \angle{-a, k}$:
    the primary rank prevents preemption,
    and the secondary rank prioritizes the classes
    when starting a new job.
  \item
    Preemptive priority has rank function
    $r(k, a) = \angle{k, -a}$:
    because~$k$ is the primary rank,
    jobs from high-priority classes preempt those in low priority classes.
  \end{itemize}
\end{example}

\begin{example}
  \label{ex:soap_srpt}
  The \emph{shortest job first}~(SJF),
  \emph{preemptive shortest job first}~(PSJF),
  and \emph{shortest remaining processing time}~(SRPT) policies
  are SOAP policies.
  All three policies assume exact size information is known
  and use it when scheduling,
  so all use $\mc{D} = \R_{\geq 0}$.
  \begin{itemize}
  \item
    SJF has rank function $r(x, a) = \angle{-a, x}$:
    it is a \emph{nonpreemptive} priority policy with size as priority.
  \item
    PSJF has rank function $r(x, a) = \angle{x, -a}$:
    it is a \emph{preemptive} priority policy with size as priority.
  \item
    SRPT has rank function $r(x, a) = x - a$:
    a job's rank is its remaining size.
  \end{itemize}
\end{example}

\subsection{Newly Analyzed SOAP Policies}
\label{sub:examples_new}

\begin{example}
  \label{ex:soap_serpt}
  The \emph{shortest expected processing time}~(SEPT),
  \emph{preemptive shortest expected processing time}~(PSEPT)
  and \emph{shortest expected remaining processing time}~(SERPT) policies
  are SOAP policies.
  The policies are respective analogues of SJF, PSJF, and SRPT,
  but they do not have access to exact size information.
  \begin{itemize}
  \item
    SEPT has rank function $r(d, a) = \angle{-a, \E{X_d}}$:
    it is a \emph{nonpreemptive} priority policy
    with expected size as priority.
  \item
    PSEPT has rank function $r(d, a) = \angle{\E{X_d}, -a}$:
    it is a \emph{preemptive} priority policy with expected size as priority.
  \item
    SERPT has rank function $r(d, a) = \E[X_d > a]{X_d - a}$:
    a job's rank is its expected remaining size.
  \end{itemize}
  While SEPT and PSEPT have analyses similar to those of SJF and PSJF,
  respectively,
  SERPT has never been analyzed before in full generality.
  We have left the set of descriptors~$\mc{D}$ unspecified
  because the definitions above work for any set of descriptors.

  For concreteness, consider a system where
  all jobs have the same descriptor $\dnil$
  and size either $2$ or $14$, each with probability $1/2$.
  The resulting rank function for SERPT, shown in \fref{fig:rank_serpt},
  is \emph{nonmonotonic} with respect to age.
  This differs from
  \emph{every policy described in \fref{sub:examples_previous}},
  all of which have monotonic rank functions.
  The potential nonmonotonicity of SERPT's rank function
  has prevented previous techniques from analyzing SERPT in full generality.
  We give the first response time analysis of SERPT
  using our general analysis of all SOAP policies
  (\fref{sub:gittins_serpt}).

  \begin{figure}
    \centering
    \begin{tikzpicture}[rank plot]
  \xguide{2}{12}
  \xguide{14}{0}
  \yguide{2}{6}
  \yguide{0}{8}
  \yguide{2}{12}

  \axes{14.8}{14}{$a$}{$r(\emptyset, a)$}

  \draw[primary] (0, 8) -- (2, 6);
  \draw[primary] (2, 12) -- (14, 0);

  \jump{2}{6}{12}
\end{tikzpicture}

    \captiondetail
    The rank function for SERPT
    using the distribution described in \fref{ex:soap_serpt}:
    jobs have size either~$2$ or~$14$, each with probability~$1/2$.
    The rank is the expected remaining size of a job
    given it has reached its age~$a$.
    In this case, the initial expected size is $8$,
    but if the job does not finish at age~$2$,
    then we know it must be size~$14$,
    so its expected remaining size jumps up to $12$.

    \caption{Rank Function for SERPT (\Fref{ex:soap_serpt})}
    \label{fig:rank_serpt}
  \end{figure}
\end{example}

\begin{example}
  \label{ex:soap_gittins}
  The \emph{Gittins index} of a job with descriptor~$d$ and age~$a$ is
  \citep{book_gittins, m/g/1_gittins_aalto}
  \begin{equation*}
    G(d, a)
    = \sup_{\Delta > 0} \frac{%
        \P[X_d > a]{X_d - a \leq \Delta}}{%
        \E[X_d > a]{\min\{X_d - a, \Delta\}}}
    = \sup_{\Delta > 0} \frac{%
        \int_a^{a + \Delta} \density[d]{t} \, dt}{%
        \int_a^{a + \Delta} \tail[d]{t} \, dt},
  \end{equation*}
  where $\density[d]{}$ and $\tail[d]{}$
  are the density and tail functions of~$X_d$, respectively.
  The \emph{Gittins index policy} is the scheduling policy
  that always serves the job of maximal Gittins index,
  and it is known to minimize mean response time
  in the M/G/1 queue~\citep{book_gittins}.
  Although optimality of the Gittins index policy has long been known,
  only a few special cases have been analyzed in the past
  \citep{multiclass_ayesta, rs_slowdown_hyytia}.
  The Gittins index policy is a SOAP policy with rank function
  \begin{equation*}
    r(d, a) = \frac{1}{G(d, a)}.
  \end{equation*}
  Like the policies in \fref{ex:soap_serpt},
  the Gittins index policy can be defined with any set of descriptors.

  The Gittins index policy is in general not the same as SERPT,
  as shown in \fref{fig:rank_gittins},
  but, like SERPT,
  the Gittins index policy often uses a nonmonotonic rank function,
  making it impossible to analyze in general using previous techniques.
  We give the first response time analysis of the Gittins index policy
  using our general analysis of all SOAP policies (\fref{sub:gittins_serpt}).

  \begin{figure}
    \centering
    \begin{tikzpicture}[rank plot]
  \xguide{2}{12}
  \xguide{14}{0}
  \yguide{0}{4}
  \yguide{2}{12}

  \axes{14.8}{14}{$a$}{$r(\emptyset, a)$}

  \draw[primary] (0,4) -- (2, 0);
  \draw[primary] (2,12) -- (14, 0);

  \jump{2}{0}{12}
\end{tikzpicture}

    \captiondetail
    The rank function for the Gittins Index Policy
    using the same distribution as in \fref{fig:rank_serpt}:
    jobs have size either~$2$ or~$14$, each with probability~$1/2$.
    Compared to SERPT,
    the Gittins index policy gives more priority to jobs
    before they reach age~$2$.
    For instance,
    while SERPT ranks a job with age~$1.99$ on par with
    a hypothetical job that deterministically has remaining size~$6.01$,
    the Gittins index policy ranks such a job on par with
    a hypothetical job that deterministically has remaining size~$0.02$.
    This reflects the fact that
    it is almost free to run such a job to age~$2$,
    just in case it is about to finish.
    We show in \fref{sub:gittins_serpt}
    that the Gittins index policy achieves lower mean response time than SERPT
    due to its prioritizing potentially short jobs.

    \caption{Rank Function for Gittins Index (\Fref{ex:soap_gittins})}
    \label{fig:rank_gittins}
  \end{figure}
\end{example}

\begin{example}
  \label{ex:soap_checkpoints}
  Consider a system in which jobs,
  rather than being completely nonpreemptible or preemptible,
  are \emph{preemptible at specific checkpoints},
  say every $1$~time unit.
  The \emph{discretized~FB} policy is a variant of FB
  for jobs with checkpoints:
  when possible, it serves the job of minimal age,
  but it does not preempt jobs between checkpoints\footnote{%
    We note that it is possible to model discretized~FB
    as an MLPS policy with infinitely many thresholds.}.
  Discretized~FB is a SOAP policy.
  It uses no static information, so $\mc{D} = \{\dnil\}$,
  and it has rank function
  \begin{equation*}
    r(\dnil, a) = \angle{\floor{a} - a, a},
  \end{equation*}
  This rank function is illustrated in \fref{fig:rank_checkpoints}.
  Roughly speaking, the primary rank encodes the ``discretized'' aspect,
  preempting a job only at integer ages~$a$ when $\floor{a} - a = 0$,
  and the secondary rank encodes the ``FB'' aspect.

  Discretized~FB is just one example of a policy
  for jobs preemptible only at specific checkpoints,
  but we can ``discretize'' any other SOAP policy
  by using primary rank $\floor{a} - a$.
  For instance, \emph{discretized SRPT} has rank function
  $r(x, a) = \angle{\floor{a} - a, x - a}$.

  \begin{figure}
    \centering
    \begin{tikzpicture}[rank plot]
  \begin{scope}
    \downmode

    \xguide[$1$]{2}{-6}
    \xguide[$2$]{4}{-6}
    \xguide[$3$]{6}{-6}
    \xguide[$4$]{8}{-6}
    \xguide[$5$]{10}{-6}
    \xguide[$6$]{12}{-6}

    \yguide[$-1$]{14}{-6}

    \axes{14.8}{-7}{$a$}{$r_1(\emptyset, a)$}

    \draw[primary] (0, 0) -- (2, -6);
    \draw[primary] (2, 0) -- (4, -6);
    \draw[primary] (4, 0) -- (6, -6);
    \draw[primary] (6, 0) -- (8, -6);
    \draw[primary] (8, 0) -- (10, -6);
    \draw[primary] (10, 0) -- (12, -6);
    \draw[primary] (12, 0) -- (12.5, -6*0.5/2);
    \draw[primary, dotted] (12.5, -6*0.5/2) -- (13.7, -6*1.7/2);

    \jump{2}{-6}{0}
    \jump{4}{-6}{0}
    \jump{6}{-6}{0}
    \jump{8}{-6}{0}
    \jump{10}{-6}{0}
    \jump{12}{-6}{0}
  \end{scope}

  \begin{scope}[shift={(0, -30)}]
    \xguide[$b$]{7}{6}
    \yguide[$b$]{7}{6}

    \axes{14.8}{14}{$a$}{$r_2(\emptyset, a)$}

    \draw[secondary] (0,0) -- (12.5, 12*12.5/14);
    \draw[secondary, dotted] (12.5, 12*12.5/14) -- (13.7, 12*13.7/14);
  \end{scope}
\end{tikzpicture}

    \captiondetail
    The rank function for the discretized~FB policy described in
    \fref{ex:soap_checkpoints}.
    The primary rank~(top)
    ensures that jobs are only preempted at integer ages.
    This is because new jobs enter the system with primary rank~$0$,
    and jobs only have primary rank~$0$ at integer ages.
    The secondary rank~(bottom)
    prioritizes the job of lowest age, as in traditional~FB.

    \caption{Rank Function for Discretized~FB (\Fref{ex:soap_checkpoints})}
    \label{fig:rank_checkpoints}
  \end{figure}
\end{example}

We have seen a variety of features that SOAP policies can model:
\begin{itemize}
\item
  jobs that are nonpreemptible, preemptible,
  or preemptible at checkpoints;
\item
  jobs with known or unknown exact size;
\item
  priority based on a job's exact size or expected size;
\item
  class-based priority in multiclass systems; and
\item
  priority that changes nonmonotonically as a job ages.
\end{itemize}
As the following examples show, SOAP policies go even further:
they can \emph{combine many such features} as part of a single policy.

\begin{example}
  \label{ex:soap_humans_and_robots}
  Consider a system with two customer classes, humans~($H$) and robots~($R$).
  \begin{itemize}
  \item
    Humans, unpredictable and easily offended,
    have unknown service time, are nonpreemptible,
    and are served according to FCFS relative to other humans.
  \item
    Robots, precise and ruthlessly efficient,
    have known service time, are preemptible,
    and are served according to SRPT relative to other robots.
  \end{itemize}
  We can model the system using
  $\mc{D} = \{[H, ?]\} \cup \{[R, x] \mid x \in \R_{\geq 0}\}$,
  where $?$ denotes unknown size.
  A reasonable policy might have humans outrank most robots
  but let short robots,
  say those with remaining size less than some threshold $x_H$,
  outrank humans that have not yet started service.
  This results in rank function
  \begin{align*}
    r([H, ?], a) &= \angle{-a, x_H} \\
    r([R, x], a) &= \angle{0, x - a},
  \end{align*}
  which uses primary rank to encode preemptibility
  and secondary rank to encode priority.
  We analyze this system in \fref{sub:humans_and_robots}.
\end{example}

\begin{example}
  Consider a system as in \fref{ex:soap_classes} with $n = 3$ classes.
  Suppose that, in addition to class-based priority,
  \begin{itemize}
  \item
    jobs in class~$1$ are preemptible,
    have known size, and are served according to SRPT;
  \item
    jobs in class~$2$ are nonpreemptible,
    have known size, and are served according to SJF; and
  \item
    jobs in class~$3$ are preemptible at specific checkpoints,
    have unknown size,
    and are served according to discretized~FB,
    as in \fref{ex:soap_checkpoints}.
  \end{itemize}
  The static information of a job is its class and, if known,
  its size, so $\mc{D} = \{1, 2\} \times \R_{\geq 0} \cup \{[3, ?]\}$,
  where $?$ denotes unknown size.
  The rank function, which uses $\mc{R} = \R^3$ as the set of ranks, is
  \begin{align*}
    r([1, x], a) &= \angle{0, 1, x - a} \\
    r([2, x], a) &= \angle{-a, 2, x} \\
    r([3, ?], a) &= \angle{\floor{a} - a, 3, a}.
  \end{align*}
  The components of the rank respectively encode
  preemptibility, class-based priority, and the policy used within each class.
\end{example}

\subsection{SOAP Policies with LCFS Tiebreaking}

There are two common SOAP policies that
use LCFS tiebreaking instead of FCFS.
The \emph{last-come, first-serve}~(LCFS) policy,
which has the same rank function as FCFS in \fref{ex:soap_fcfs}
but uses LCFS tiebreaking.
The rank function still ensures nonpreemption,
but now ties between jobs of age~$0$ are broken by LCFS.
Similarly, the \emph{preemptive last-come, first-serve}~(PLCFS) policy
is a SOAP policy with constant rank function and LCFS tiebreaking.
SOAP policies that use LCFS tiebreaking admit essentially the same analysis
as those that use FCFS tiebreaking (\fref{app:lcfs}).

\subsection{Non-SOAP Policies}

As our examples have demonstrated,
there is an extremely wide variety of SOAP policies.
However, there are some policies which are not SOAP policies,
many of which fit into three broad categories.

First, some policies \emph{cannot be expressed using descriptors
  that are distributed i.i.d. for each arriving job}.
For example,
the \emph{earliest deadline first}~(EDF) policy could be a SOAP policy
if each job's descriptor were its deadline,
but deadlines cannot be i.i.d. because later arrivals need later deadlines.

Second, some policies
\emph{require a tiebreaking rule other than FCFS or LCFS}.
For example,
the \emph{random order of service}~(ROS) policy
could be a SOAP policy using the rank function in \fref{ex:soap_fcfs}
if it could break ties between jobs of age~$0$ differently.
A future generalization of the SOAP class might allow for ROS tiebreaking
because, like FCFS and LCFS, it serves one job at a time.
In contrast, the \emph{processor sharing}~(PS) policy,
which is also not a SOAP policy,
requires a fundamentally different tiebreaking rule.

Third, some policies have \emph{job priorities that are context-dependent}.
For example, a nonpreemptive policy in a multiclass system that
tries to alternate between serving jobs of class~$1$ and class~$2$
is not a SOAP policy,
because the priority of a job depends on external context,
namely the class of the previously served job.
The rank function approach used by SOAP policies
inherently considers each job individually,
so there is no way to capture such context.

\section{How to Handle Any Rank Function}
\label{sec:key_ideas}

We have seen how to express a vast space of policies in the SOAP framework
by careful choice of rank function.
However, as demonstrated by \fref{sub:examples_new},
the rank function that encodes a SOAP policy can be very complicated.
This leaves us with a difficult technical challenge:
how do we analyze SOAP policies with arbitrary rank functions?

Before tackling arbitrary rank functions,
let us recall how classic response time analyses work.
Though there are of course many approaches,
all of the policies in \fref{sub:examples_previous} can be analyzed with the
``tagged job'' approach,
which follows a particular job through the system
to analyze its response time.
For instance, consider tagging a job of size~$x$
in a system using PSJF (\fref{ex:soap_srpt}).
There are two types of jobs that outrank the tagged job:
\begin{itemize}
\item
  jobs of size at most~$x$
  that are present in the system when the tagged job arrives and
\item
  jobs of size less than~$x$ that arrive at the system
  after the tagged job arrives but before it completes.
\end{itemize}

One way to think about PSJF is to view the tagged job as
seeing the system through ``transformer glasses'' \citep{book_harchol-balter}
which transform the system by
\emph{hiding jobs that the tagged job outranks}.
For PSJF, this transformation is simple because
each job's rank is essentially constant\footnote{%
  Due to FCFS tiebreaking,
  we could encode PSJF without its secondary rank.}.
A similar approach still works
for policies with increasing or decreasing rank functions,
but the hiding transformation becomes more complicated.
For instance, under SRPT (\fref{ex:soap_srpt}),
which has a decreasing rank function,
a tagged job of size~$x$ sees a system transformed as follows.
\begin{itemize}
\item
  Jobs that arrive after the tagged job are hidden
  if their size is at least $x - a$,
  where $a$ is the tagged job's age when the other job arrives.
\item
  Jobs that arrive before the tagged job
  are hidden if their \emph{remaining} size is greater than~$x$.
  It may be that a job of \emph{initial} size greater than~$x$
  remains visible.
\end{itemize}
In more general terms:
because jobs' ranks change with age,
the tagged job's hiding criterion changes with its age,
and whether or not other jobs satisfy that criterion changes with their ages.
Handling these changes in rank is already tricky for SRPT,
where a job's rank only decreases with age.
The situation becomes even more complex
when working with \emph{nonmonotonic} rank functions.

\subsection{Conventions}

For the remainder of this section,
we examine the journey of a tagged job~$J$ through the system.
The tagged job~$J$ has descriptor~$d_J$ and size~$x_J$.
Note that we may use $J$'s size as part of our analysis
even if the scheduler does not have access to exact job sizes.

Throughout,
we call jobs other than~$J$ \emph{old} if they arrive before $J$
and \emph{new} if they arrive after~$J$.
As a mnemonic,
we name old jobs~$I$ and name new jobs~$K$,
using subscripts when there are multiple such jobs.
In examples, old jobs have descriptor~$d_I$ and new jobs have descriptor~$d_K$,
though in a real system it may of course happen that
different old or new jobs have different descriptors.
When it is unspecified whether a job is new or old, we name the job~$L$.

\subsection{Nonmonotonicity Difficulties}
\label{sub:nonmonotonicity}

There are two major obstacles to analyzing policies with
arbitrary nonmonotonic rank functions
that do not occur when analyzing SRPT,
which has a monotonic rank function.
We illustrate these obstacles below and in \fref{fig:rank_nonmonotonic}.
The first obstacle concerns the nonmonotonicity of $J$'s rank.
\begin{itemize}
\item
  In SRPT, we \emph{permanently} hide some other jobs
  based on $J$'s current rank.
\item
  In general, another job~$L$ might be only \emph{temporarily} hidden.
  If $J$'s rank starts below $L$'s rank but later exceeds it,
  the initially hidden~$L$ becomes visible again.
\end{itemize}
The second obstacle concerns the nonmonotonicity of the ranks of old jobs.
\begin{itemize}
\item
  In SRPT, an old job \emph{permanently} outranks $J$
  if served for long enough before $J$ arrives.
\item
  In general, an old job~$I$ might only \emph{temporarily} outrank~$J$.
  Furthermore, if $I$'s rank oscillates above and below $J$'s initial rank,
  whether $I$ gets hidden or stays visible depends on
  \emph{when during $I$'s service $J$ arrives}.
\end{itemize}

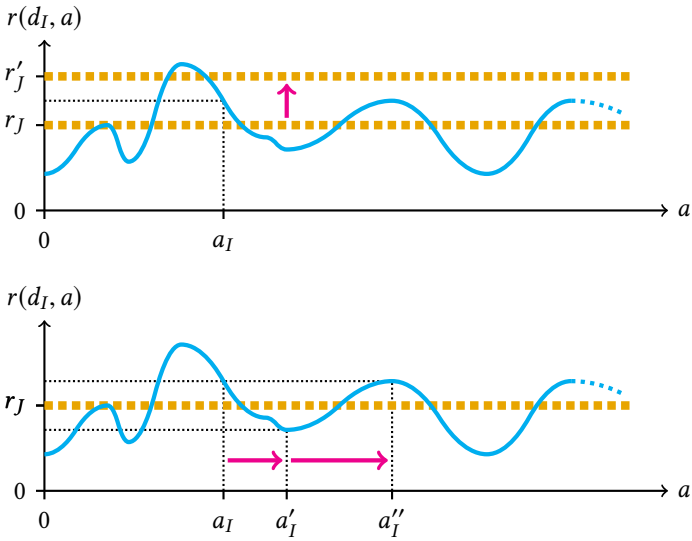
\begin{figure}
  \centering
  \begin{tikzpicture}[rank plot]
  \begin{scope}
    \xguide[$a_I$]{4.25}{9}
    \yguide[$r_J$]{0}{7}
    \yguidenolabel{4.25}{9}
    \yguide[$r_J'$]{0}{11}

    \axes{14.8}{14}{$a$}{$r(d_I, a)$}

    \draw[cutoff] (0, 7) -- (14, 7);
    \draw[cutoff] (0, 11) -- (14, 11);
    \Isnake

    \draw[rworst, ultra thick, solid, ->] (5.75, 7.6) -- (5.75, 10.4);
  \end{scope}

  \begin{scope}[shift={(0, -23)}]
    \xguide[$a_I$]{4.25}{9}
    \xguide[$a_I'$]{5.75}{5}
    \xguide[$a_I''$]{8.25}{9}
    \yguide[$r_J$]{0}{7}
    \yguidenolabel{5.75}{5}
    \yguide[$r_J$]{0}{7}
    \yguidenolabel{8.25}{9}

    \axes{14.8}{14}{$a$}{$r(d_I, a)$}

    \draw[rworst, ultra thick, solid, ->] (4.35, 2.5) -- (5.65, 2.5);
    \draw[rworst, ultra thick, solid, ->] (5.85, 2.5) -- (8.15, 2.5);

    \draw[cutoff] (0, 7) -- (14, 7);
    \Isnake
  \end{scope}
\end{tikzpicture}

  \captiondetail
  We show two obstacles to the ``transformer glasses'' approach
  by examining the interaction between
  the tagged job~$J$ and an old job~$I$ of age~$a_I$.
  The cyan curve shows $I$'s rank as a function of its age.
  First (top), suppose a tagged job~$J$
  has rank~$r_J$ upon entering the system,
  and suppose that old job~$I$ of age~$a_I$ is already in the system.
  Consider how $J$ views $I$ if, as pictured, $r_J \prec r(d_I, a_I)$.
  Then $J$ outranks~$I$, so $I$ is initially hidden.
  However, $I$ may be only \emph{temporarily} hidden,
  because $J$'s rank may later increase to $r_J' \succ r(d_I, a_I)$.
  Second (bottom), consider the same jobs~$J$ and $I$
  and suppose $J$ has not entered the system yet.
  Whether or not $I$ will be hidden from $J$
  depends on \emph{$I$'s age when $J$ arrives}.
  For instance, as $I$ advances in age from $a_I$ to $a_I'$ and later~$a_I''$,
  it switches back and forth between being visible to and hidden from $J$.

  \caption{Difficulties with Nonmonotonic Rank Functions}
  \label{fig:rank_nonmonotonic}
\end{figure}

Dealing with such arbitrarily varying rank functions appears intractable.
We need \emph{two key insights} in order to handle nonmonotonic rank functions:
the \emph{Pessimism Principle} (\fref{sub:pessimism_principle})
and the \emph{Vacation Transformation} (\fref{sub:vacation_transformation}).

\subsection{The Pessimism Principle}
\label{sub:pessimism_principle}

We call the amount of time any job~$L$ is served before $J$ completes
\emph{$J$'s delay due to~$L\esub$}.
The response time of $J$ is its size~$x_J$
plus its delays due to all jobs that are in the system with $J$ at some point.

Suppose that a new job~$K$ arrives when $J$ has age~$a_J$.
To analyze $J$'s response time, we have to know its delay due to~$K$.
As we saw in \fref{sub:nonmonotonicity},
deciding whether or not to hide $K$
based only on $J$'s current rank $r(d_J, a_J)$ will not work.
Instead, we need to examine $J$'s \emph{current and future ranks}.

\begin{definition}
  \label{def:rworst}
  The \emph{worst future rank} of a job
  with descriptor~$d$, size~$x$, and age~$a$ is
  \begin{equation*}
    \rworst[d, x]{a} = \sup_{\mathclap{a \leq b < x}} r(d, b).
  \end{equation*}
  Note that the worst future rank only considers ages up to the job's size~$x$.
  See \fref{fig:rworst} for an illustration.

  \begin{figure}
    \centering
    \begin{tikzpicture}[rank plot]
  \xguide[$x$]{10.25}{5}

  \axes{14.8}{14}{$a$}{$r(d, a)$}

  \Jsnake

  \draw[rworst]
  (0, 12) -- (5.5, 12)
  -- \snakefirst{(5.5, 12)}{(7, 4)}
  -- \snakefirst{(8, 8)}{(9.5, 2)}
  -- (10.25, 5);
  \point[rworst point]{(10.25, 5)}
  \node[above] at (3.5, 12) {$\rworst[d, x]{a}$};
\end{tikzpicture}

    \captiondetail
    The relationship between rank $r(d, a)$ (solid cyan),
    and worst future rank $\rworst[d, x]{a}$ (dashed magenta)
    for a job with descriptor~$d$ and size~$x$.

    \caption{Illustration of Worst Future Rank (\Fref{def:rworst})}
    \label{fig:rworst}
  \end{figure}
\end{definition}

The reason we care about $J$'s worst future rank is that
for the purposes of computing $J$'s delay delay due to~$K$,
we can essentially pretend that $J$ has its worst future rank
\begin{equation*}
  \rJw = \rworst[d_J, x_J]{a_J},
\end{equation*}
as illustrated in \fref{fig:peak}.
This is because before $J$ reaches its worst future rank,
$K$ must either complete or surpass\footnote{%
  Unless otherwise noted, we mean ``surpass'' in a weak sense.
  That is, to surpass a rank~$r$ means to attain rank at least~$r$.}
rank~$\rJw$.

\begin{figure}
  \centering
  \begin{tikzpicture}[rank plot]
  \xguide[$a_J$]{1.5}{3}
  \xguide[$a_J'$]{4}{8.5}
  \xguide[$a_J''$]{5.5}{12}
  \yguide[$r_K$]{0}{8.5}
  \yguide[$\rJw$]{5.5}{12}

  \axes{14.8}{14}{$a$}{$r(d_J, a)$}

  \draw[cutoff, name path=cutoff] (0, 8.5) -- (14, 8.5);
  \Jsnake
\end{tikzpicture}

  \captiondetail
  A new job~$K$ arrives when $J$,
  whose rank as a function of age is shown in cyan,
  has age~$a_J$.
  Suppose for simplicity that $K$'s rank is constant at~$r_K$.
  An instant after $J$ reaches age~$a_J'$, it will be outranked by $K$,
  so $K$ will complete before~$J$ reaches age~$\aJw$.
  In practice, $K$'s rank may change,
  but $K$ will still complete before~$J$ reaches age~$\aJw$
  unless $K$ surpasses rank~$\rJw$,
  in which case $K$ never outranks $J$ again.
  Thus, for the purposes of finding $J$'s delay due to $K$,
  we can pretend that $J$'s rank is $\rJw$
  for all of $J$'s ages before~$\aJw$.

  \caption{A Tagged Job's Delay due to a New Job}
  \label{fig:peak}
\end{figure}

We have so far focused on a new job~$K$,
but the story is very similar for old jobs.
The result is the Pessimism Principle,
so named for its pessimistic focus on the worst future rank.

\begin{namedobservation*}{Pessimism Principle}
  The tagged job~$J$'s delay due to any other job~$L$
  is the amount of time $L$ is served until it either
  \emph{completes} or \emph{surpasses $J$'s worst future rank}.
  To be precise, letting
  \begin{equation*}
    \rJw(a) = \rworst[d_J, x_J]{a},
  \end{equation*}
  this means the following.
  \begin{itemize}
  \item
    Each old job~$I$ is served until
    it completes or first has rank $r_I \succ \rJw(0)$.
    In particular, if $r_I = \rJw(0)$,
    then $I$ outranks~$J$ due to FCFS tiebreaking,
    thus the strict inequality\footnote{%
      See \fref{app:rworst} for discussion of corner cases
      where we need $\succeq$ instead of $\succ$.}.
  \item
    Each new job~$K$ is served until
    it completes or first has rank $r_K \succeq \rJw(a_J)$,
    where $a_J$ is the age of $J$ when $K$ arrives.
  \end{itemize}
  Furthermore, in both cases, the service occurs
  \emph{before $J$ is served while at its worst future rank}.

  To clarify, the discussion above addresses the
  \emph{total amount of service another job receives}
  while $J$ is in the system.
  The service need not be contiguous
  but might be interleaved with that of $J$ and other jobs.
\end{namedobservation*}

The Pessimism Principle gives an implicit description of
$J$'s delay due to any other job~$L$,
but it remains to explicitly find this delay's distribution.
We do this now for the case where $L$ is new,
treating the case where $L$ is old in \fref{sub:vacation_transformation}.

\begin{definition}
  \label{def:new_work}
  Let $r$ be a rank.
  The \emph{new $r$-work} is a random variable, written $\Xnew{r}$,
  representing how long a job that just arrived to the system
  is served until it completes or surpasses rank~$r$.
  Specifically, we define $\Xnew{r} = \Xnew[D]{r}$,
  where $D$ is the random descriptor assigned to a new job and,
  for any specific descriptor~$d$,
  \begin{align*}
    c_d[r] &= \inf\{a \geq 0 \mid r(d, a) \succeq r\} \\
    \Xnew[d]{r} &= \min\{X_d, c_d[r]\}.
  \end{align*}
  That is, $c_d[r]$ is the \emph{cutoff age} at which
  a new job with descriptor~$d$ surpasses rank~$r$.
  See \fref{fig:new_work} for an illustration.

  \begin{figure}
    \centering
    \begin{tikzpicture}[rank plot]
  \interval{original}{(0, 10)}{(6, 10)}{new}

  \xguide[${c_d[r]}$]{6}{10}
  \yguide[$r$]{0}{10}

  \axes{14.8}{14}{$a$}{$r(d, a)$}

  \draw[cutoff] (0, 10) -- (14, 10);
  \Ksnake
\end{tikzpicture}

    \captiondetail
    The new $r$-work for a job with descriptor~$d$,
    whose rank as a function of age is shown in cyan,
    is the amount of service the job requires
    until it either completes or surpasses rank~$r$
    at the cutoff age~$c_d[r]$.
    Pictorially, the new $r$-work is the amount of service the job requires
    while its age is in the green region.
    The new $r$-work is at most~$c_d[r]$ but may be less if the job completes.
    Note that the new $r$-work is not impacted by the job's rank
    at ages $a > c_d[r]$,
    even if $r(d, a) \prec r$.

    \caption{Illustration of New Work (\Fref{def:new_work})}
    \label{fig:new_work}
  \end{figure}
\end{definition}

Together, the Pessimism Principle and \fref{def:new_work} say that
if a new job~$K$ has random descriptor and arrives when $J$ has age~$a_J$,
then $J$'s delay due to $K$ is $\Xnew{\rworst[d_J, x_J]{a_J}}$.

\begin{example}
  \label{ex:new_work_srpt}
  Consider SRPT as described in \fref{ex:soap_srpt},
  in which a job's descriptor is its size and its rank is its remaining size.
  Suppose that a new job~$K$ of size~$x_K$ arrives when $J$ has age~$a_J$.
  Under SRPT, $J$'s worst future rank is its current rank
  $r(x_J, a_J) = x_J - a_J$,
  so the cutoff age for~$K$ is
  \begin{equation*}
    c_{x_K}[x_J - a_J] =
    \begin{cases}
      \infty & \text{if } x_K < x_J - a_J \\
      0 & \text{if } x_K \geq x_J - a_J.
    \end{cases}
  \end{equation*}
  That is, $K$ will always outrank $J$ if $x_K < x_J - a_J$,
  and $K$ will never outrank $J$ if $x_K \geq x_J - a_J$.
  This means $J$'s delay due to $K$ is
  \begin{equation*}
    \Xnew[x_K]{x_J - a_J} = x_K \1(x_K < x_J - a_J),
  \end{equation*}
  where $\1$ is the indicator function.

  When analyzing the response time of SRPT,
  we need to know $J$'s delay due to a \emph{random} new job,
  meaning one with size drawn from the size distribution~$X$.
  This is the new $(x_J - a_J)$-work,
  \begin{equation*}
    \Xnew{x_J - a_J} = \Xnew[X]{x_J - a_J} = X \1(X < x_J - a_J).
  \end{equation*}
\end{example}

\begin{example}
  \label{ex:new_work_serpt}
  Consider the SERPT system described in \fref{ex:soap_serpt},
  in which \emph{all} jobs have the same descriptor~$\dnil$
  and the same two-point size distribution:
  jobs are size~$2$ with probability~$1/2$ and size~$14$ otherwise.
  The rank function is
  \begin{equation*}
    r(\dnil, a) = \E[X > a]{X - a} =
    \begin{cases}
      8 - a & \text{if } a < 2 \\
      14 - a & \text{if } a \geq 2,
    \end{cases}
  \end{equation*}
  as shown in \fref{fig:rank_serpt}.

  Suppose that a new job~$K$ arrives when $J$ has age $a_J < 2$.
  The worst future rank of~$J$ depends crucially on $J$'s size~$x_J$.
  \begin{itemize}
  \item
    If $x_J = 2$, then $J$'s worst future rank is
    its current rank, $8 - a_J$.
    In this case, the cutoff age for $K$ is $c_\dnil[8 - a_J] = 0$
    because $K$'s initial rank is $8$,
    which is at least $J$'s worst future rank $8 - a_J$.
    This means $J$'s delay due to $K$ is
    \begin{equation*}
      \Xnew{8 - a_J} = \Xnew[\dnil]{8 - a_J} = 0.
    \end{equation*}
  \item
    If instead $x_J = 14$, then $J$'s worst future rank is $12$.
    In this case, the cutoff age for $K$ is $c_\dnil[12] = 2$
    because $K$ will either complete or jump up to rank~$12$
    when it reaches age~$2$.
    This means $J$'s delay due to $K$ is
    \begin{equation*}
      \Xnew{12} = \Xnew[\dnil]{12} = 2.
    \end{equation*}
  \end{itemize}
\end{example}

\subsection{The Vacation Transformation}
\label{sub:vacation_transformation}

We have seen how the Pessimism Principle shows us
how long each new job delays the tagged job~$J$.
The question remains: how long does each old job delay~$J$?
This is much harder than the corresponding question for new jobs
because \emph{old jobs can have any age},
whereas new jobs always start at age~$0$.

Fortunately, we are actually not directly concerned with
the delay due to individual old jobs.
What ultimately matters is the delay due to \emph{all old jobs together}.
It turns out that we can view this total delay as
the \emph{queueing time of a carefully transformed system}.
The careful transformation in question is the Vacation Transformation,
but we are still a few definitions away from presenting it.

Our goal is to find the $J$'s total delay due to old jobs.
We can think of this delay as the amount of ``relevant'' work in the system
at the moment $J$ arrives.
Because Poisson arrivals see time averages \citep{pasta_wolff},
the distribution of the amount relevant work seen by $J$ upon arrival
is the \emph{stationary distribution} of the amount of relevant work.
Thus, in this section,
we imagine $J$ as a \emph{witness} to the system,
watching other jobs enter, receive service, and exit.
For this purpose, the most important fact about $J$ is
its worst future rank upon arrival,
\begin{equation*}
  \rJw = \rworst[d_J, x_J]{0}.
\end{equation*}

So far, we have considered old jobs as a monolithic category,
but it is useful to consider three subcategories.
At any moment in time, we can classify old jobs as follows.
\begin{itemize}
\item
  \emph{Discarded} old jobs currently have rank greater than $\rJw$.
\item
  \emph{Original} old jobs currently have rank at most $\rJw$
  and \emph{have always had} rank at most $\rJw$ since arriving themselves.
\item
  \emph{Recycled} old jobs currently have rank at most $\rJw$
  but had rank greater than $\rJw$ at some point in the past.
\end{itemize}
More generally,
we may call a job discarded, original, or recycled
\emph{with respect to rank~$r\esub$},
in which case we replace $\rJw$ with~$r$.

An old job~$I$ goes through the following transitions between these categories,
as shown in Figures~\ref{fig:transformer_glasses} and~\ref{fig:old_work}.
\begin{itemize}
\item
  When $I$ arrives in the system,
  if its initial rank is at most $\rJw$,
  it starts out original.
  Otherwise, it starts out discarded.
\item
  As $I$ ages, it may become discarded if it is not already.
\item
  As $I$ ages further, it may become recycled,
  then discarded again, then recycled again, and so on until it completes.
\end{itemize}

Eventually, $J$ will arrive,
and each of the old jobs will delay $J$ by some amount of time
based on their category at the moment when $J$ arrives.
By the Pessimism Principle, $J$'s delay due to discarded jobs is~$0$,
so such jobs do not concern us further.
Original jobs are similar to new jobs
but, due to the FCFS tiebreaking rule, not quite the same.
Recycled jobs are the most difficult type of old job to handle,
but fortunately, as we will soon see,
$J$ sees \emph{at most one} recycled job in the system when it arrives.

For the purposes of analyzing $J$'s response time,
we view the system through ``transformer glasses'' \citep{book_harchol-balter}
through which \emph{only original and recycled jobs are visible},
as illustrated in \fref{fig:transformer_glasses}.
In the transformed system,
jobs are transformed such that when they become discarded, they complete.
Thus, the total work in the transformed system is exactly
the ``relevant'' work in the untransformed system,
which would be $J$'s total delay due to old jobs
were $J$ to arrive immediately.
Therefore, our goal is to find the
\emph{stationary distribution of work in the transformed system}.

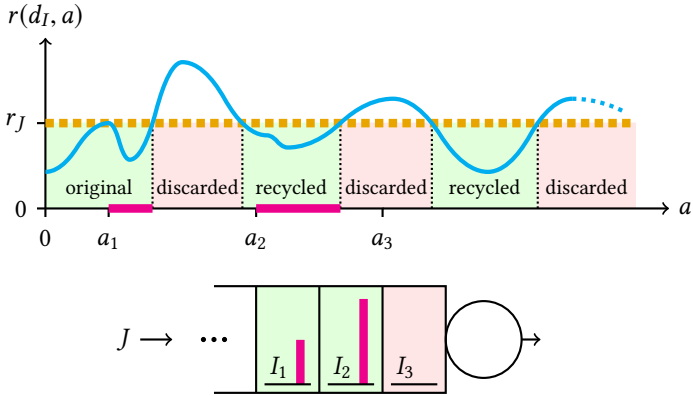
\begin{figure}
  \centering
  \begin{tikzpicture}[rank plot]
  \coordinate (a_1) at (1.5, 0);
  \coordinate (a_2) at (5, 0);
  \coordinate (a_3) at (8, 0);

  \xguidept{$a_1$}{(a_1)}
  \xguidept{$a_2$}{(a_2)}
  \xguidept{$a_3$}{(a_3)}
  \yguide[$\rJw$]{0}{7}

  \axes{14.8}{14}{$a$}{$r(d_I, a)$}

  \draw[cutoff, name path=cutoff] (0, 7) -- (14, 7);
  \Isnake
  \path[name intersections={of=rank and cutoff}];

  \coordinate (c_1) at ($ (0,0)!(intersection-2)!(1,0) $);
  \coordinate (c_2) at ($ (0,0)!(intersection-4)!(1,0) $);
  \draw[rworst, solid] (a_1) -- (c_1);
  \draw[rworst, solid] (a_2) -- (c_2);

  \begin{scope}[on background layer]
    \original{(0, 7)}{(intersection-2)}
    \discarded{(intersection-2)}{(intersection-3)}
    \recycled{(intersection-3)}{(intersection-4)}
    \discarded{(intersection-4)}{(intersection-5)}
    \recycled{(intersection-5)}{(intersection-6)}
    \discarded{(intersection-6)}{(14, 7)}

    \xguidept{}{(intersection-2)}
    \xguidept{}{(intersection-3)}
    \xguidept{}{(intersection-4)}
    \xguidept{}{(intersection-5)}
    \xguidept{}{(intersection-6)}
  \end{scope}

  \begin{scope}[shift={(0.5, -10.75)}, y=14]
    \fill[original] (4.5, -1.25) rectangle (6, 1.25);
    \fill[recycled] (6, -1.25) rectangle (7.5, 1.25);
    \fill[discarded] (7.5, -1.25) rectangle (9, 1.25);

    \draw (3.5, -1.25) -- (9, -1.25) -- (9, 1.25) -- (3.5, 1.25);
    \draw (7.5, -1.25) -- ++(0, 2.5);
    \draw (6, -1.25) -- ++(0, 2.5);
    \draw (4.5, -1.25) -- ++(0, 2.5);

    \draw (9.9, 0) circle (0.9);
    \draw[->] (10.8, 0) -- (11.25, 0);
    \filldraw (3.75, 0) circle (0.5mu);
    \filldraw (3.5, 0) circle (0.5mu);
    \filldraw (3.25, 0) circle (0.5mu);

    \draw[rworst, solid] let \p1=(a_1), \p2=(c_1) in
    (5.55, -1.05) -- ++(0, \x2-\x1);
    \draw[rworst, solid] let \p1=(a_2), \p2=(c_2) in
    (7.05, -1.05) -- ++(0, \x2-\x1);
    \draw (4.7, -1.05) -- ++(1.1, 0);
    \draw (6.2, -1.05) -- ++(1.1, 0);
    \draw (7.7, -1.05) -- ++(1.1, 0);
    \node at (5.05, -0.65) {$I_1$};
    \node at (6.55, -0.65) {$I_2$};
    \node at (8.05, -0.65) {$I_3$};

    \draw[->] (1.75, 0) node[left] {$J$} -- ++(0.76, 0);
  \end{scope}
\end{tikzpicture}
\vspace{0.5\baselineskip}

  \captiondetail
  A system contains jobs $I_1$, $I_2$, and~$I_3$
  with respective ages $a_1$, $a_2$, and~$a_3$.
  All three jobs have the same descriptor~$d_I$
  and therefore the same rank as a function of age, drawn in cyan.
  We view the system from the transformed perspective of witness~$J$,
  which has worst future rank~$\rJw$.
  By the Pessimism Principle,
  if $J$ were to arrive at the system now,
  it would only be delayed by original job~$I_1$ and recycled job~$I_2$,
  so discarded job~$I_3$ is completely hidden from~$J$.
  Furthermore, jobs $I_1$ and $I_2$ only delay $J$
  until they either complete or exceed rank~$\rJw$,
  so there are upper bounds on the delays due to each,
  shown as magenta bars.

  \caption{A Witness's Transformed View of the System}
  \label{fig:transformer_glasses}
\end{figure}

Both original and recycled jobs arrive in the transformed system.
Arrivals of original jobs correspond to arrivals to the untransformed system,
but arrivals of recycled jobs occur seemingly arbitrarily,
as they are really caused by
discarded jobs transitioning to recycled in the untransformed system.
A busy period in the transformed system
always starts with the arrival of an original or recycled job.
Arrivals of original jobs continue during the busy period,
but \emph{no more recycled jobs arrive}
for the rest of the busy period.
This is because for a recycled job to arrive in the transformed system,
a discarded job has to become recycled
by receiving service in the untransformed system.
But discarded jobs never outrank original or recycled jobs,
which are present for the entire busy period in the transformed system,
so such transitions never occur.

To analyze the amount of work in the transformed system,
we need to know how long each old job spends as an original or recycled job.
We call the amount of time an old job spends as
original with respect to rank~$r$
its \emph{$0$-old $r$-work},
and we call the amount of time it spends as
recycled for the $i$th time with respect to rank~$r$
its \emph{$i$-old $r$-work},
both of which we now define formally.

\begin{definition}
  \label{def:old_work}
  Let $r$ be a rank and $d$ be a descriptor.
  The \emph{$0$-old $r$-interval} for descriptor~$d$
  is the interval of ages during which a job of descriptor~$d$ is
  original with respect to rank~$r$.
  Specifically, the interval is $[b_{0, d}[r], c_{0, d}[r]]$, where\footnote{%
    See \fref{app:rworst} for discussion of corner cases
    where we need $\succeq$ instead of $\succ$.}
  \begin{align*}
    b_{0, d}[r] &= 0 \\
    c_{0, d}[r] &= \inf\{a \geq b_{0, d}[r] \mid r(d, a) \succ r\}.
  \end{align*}
  For $i \geq 1$,
  the \emph{$i$-old $r$-interval} for descriptor~$d$
  is the interval of ages during which the job is
  recycled with respect to rank~$r$ for the $i$th time.
  Specifically, the interval is $[b_{i, d}[r], c_{i, d}[r]]$, where
  \begin{align*}
    b_{i, d}[r] &= \inf\{a > c_{i - 1, d}[r] \mid r(d, a) \preceq r\} \\
    c_{i, d}[r] &= \inf\{a > b_{i, d}[r] \mid r(d, a) \succ r\}.
  \end{align*}
  If $b_{i, d}[r] = c_{i, d}[r] = \infty$, the interval is the empty set.
  See \fref{fig:old_work} for an illustration.

  For $i \geq 0$, the \emph{$i$-old $r$-work} is a random variable,
  written $\Xold[i]{r}$,
  representing how long a job will be served
  while its age is in its $i$-old $r$-interval.
  Specifically, we define $\Xold[i]{r} = \Xold[i, D]{r}$,
  where $D$ is the random descriptor assigned to a new job and,
  for any specific descriptor~$d$,
  \begin{equation*}
    \Xold[i, d]{r} =
    \begin{cases}
      0 & \text{if } X_d < b_{i, d}[r] \\
      X_d - b_{i, d}[r] & \text{if } b_{i, d}[r] \leq X_d < c_{i, d}[r] \\
      c_{i, d}[r] - b_{i, d}[r] & \text{if } c_{i, d}[r] \leq X_d.
    \end{cases}
  \end{equation*}
  If $b_{i, d}[r] = c_{i, d}[r] = \infty$, we define $\Xold[i, d]{r} = 0$.

  \begin{figure}
    \centering
    \begin{tikzpicture}[rank plot]
  \yguide[$r$]{0}{7}

  \axes{14.8}{14}{$a$}{$r(d, a)$}

  \draw[cutoff, name path=cutoff] (0, 7) -- (14, 7);
  \Isnake
  \path[name intersections={of=rank and cutoff}];

  \begin{scope}[on background layer]
    \original{(0, 7)}{(intersection-2)}
    \discarded{(intersection-2)}{(intersection-3)}
    \recycled[1]{(intersection-3)}{(intersection-4)}
    \discarded{(intersection-4)}{(intersection-5)}
    \recycled[2]{(intersection-5)}{(intersection-6)}
    \discarded{(intersection-6)}{(14, 7)}

    \xguidept{$c_{0, d}[r]$}{(intersection-2)}
    \xguidept{$b_{1, d}[r]$}{(intersection-3)}
    \xguidept{$c_{1, d}[r]$}{(intersection-4)}
    \xguidept{$b_{2, d}[r]$}{(intersection-5)}
    \xguidept{$c_{2, d}[r]$}{(intersection-6)}
  \end{scope}
\end{tikzpicture}

    \captiondetail
    As a job ages, it can transition repeatedly
    between being discarded and recycled with respect to a rank~$r$.
    The $i$-old $r$-interval for descriptor~$d$
    is the interval $[b_{i, d}[r], c_{i, d}[r]]$ during which
    a job of descriptor~$d$ is original ($i = 0$)
    or recycled for the $i$th time ($i \geq 1$).
    We highlight the $i$-old $r$-intervals in green.
    The $i$-old $r$-work is the amount of service the job requires
    while its age is in its $i$-old $r$-interval.

    \caption{Illustration of Old Work (\Fref{def:old_work})}
    \label{fig:old_work}
  \end{figure}
\end{definition}

Suppose old job~$I$ has a random descriptor.
In the transformed system, $I$ receives service
\begin{itemize}
\item
  for time $\Xold[0]{\rJw}$ as an original job and,
\item
  for all $i \geq 1$,
  for time $\Xold[i]{\rJw}$ as a job being recycled for the $i$th time.
\end{itemize}
Note that $\Xold[0]{\rJw}$ may be~$0$,
representing $I$ starting out discarded,
and $\Xold[i]{\rJw}$ for $i \geq 1$ may be~$0$,
representing $I$ completing before being recycled for the $i$th time.

\begin{example}
  \label{ex:old_work_srpt}
  Consider SRPT as described in \fref{ex:soap_srpt},
  in which a job's descriptor is its size and its rank is its remaining size.
  Suppose that $J$ witnesses an old job~$I$ of initial size~$x_I$.
  Under SRPT, every job's rank is strictly decreasing with age,
  so $J$'s worst future rank is its initial size~$x_J$.
  The amount of time $I$ spends as an original or recycled job
  depends on its size relative to~$J$'s.
  \begin{itemize}
  \item
    If $x_I \leq x_J$, then $I$ is original until its completion
    because its rank never exceeds~$x_J$, so
    \begin{equation*}
      \Xold[0, x_I]{x_J} = x_I.
    \end{equation*}
    $I$ is never recycled, so $\Xold[i]{x_J} = 0$ for $i \geq 1$.
  \item
    If $x_I > x_J$, then $I$ is starts out discarded
    but becomes recycled at age $x_I - x_J$,
    at which point it has remaining size~$x_J$, so
    \begin{align*}
      \Xold[0, x_I]{x_J} &= 0 \\
      \Xold[1, x_I]{x_J} &= x_J.
    \end{align*}
    $I$ is recycled only once, so $\Xold[i, x_I]{x_J} = 0$ for $i \geq 2$.
  \end{itemize}

  When analyzing the response time of SRPT,
  we need to know the amount of time a \emph{random} old job,
  meaning one with size drawn from the size distribution~$X$,
  spends as an original or recycled job.
  From the above casework, we obtain
  \begin{align*}
    \Xold[0]{x_J} = \Xold[0, X]{x_J} &= X\1(X \leq x_J) \\
    \Xold[1]{x_J} = \Xold[1, X]{x_J} &= x_J\1(X > x_J),
  \end{align*}
  where $\1$ is the indicator function,
  and $\Xold[i]{x_J} = 0$ for $i \geq 2$.
\end{example}

The Vacation Transformation gives a simple specification of
the amount of work the witness~$J$ sees in the transformed system.
It follows from three observations.
\begin{itemize}
\item
  The amount of work in the transformed system
  is independent of the scheduling policy used on transformed jobs,
  provided it is work-conserving,
  so we may assume FCFS among original and recycled jobs\footnote{%
    Recall that an original or recycled job completes the transformed system
    if it either completes \emph{or becomes discarded}
    in the untransformed sysytem.}.
\item
  Because Poisson arrivals see time averages \citep{pasta_wolff},
  the stationary amount of work in the transformed system
  is the same as the stationary FCFS queueing time of an original job
  in the transformed system.
\item
  In the transformed system,
  the arrival process of recycled jobs is not Poisson,
  but they only appear at the starts of busy periods,
  so it is convenient to view them as server vacations.
\end{itemize}

\begin{namedobservation*}{Vacation Transformation}
  Consider the tagged job~$J$ with descriptor~$d_J$ and size~$x_J$
  arriving to the system,
  and let
  \begin{equation*}
    \rJw = \rworst[d_J, x_J]{0}.
  \end{equation*}
  $J$'s total delay due to old jobs
  has the same distribution as
  \emph{queueing time in a transformed M/G/1/FCFS system
    with ``sparse'' server vacations}.
  We call the transformed system the \emph{Vacation Transformation System},
  or simply \emph{VT~System}.
  In the VT~System,
  jobs arrive at rate~$\lambda$ and have size distribution $\Xold[0]{\rJw}$,
  and several types of vacations occasionally occur, as described below.

  The server in the VT~System is always in one of three states:
  \begin{itemize}
  \item
    \emph{busy}, meaning a job is in service;
  \item
    \emph{idle}, meaning the system is empty
    but the server is ready to start a job immediately should one arrive; or
  \item
    \emph{on vacation}, meaning the server will not serve jobs
    until the vacation finishes,
    even if there are jobs in the system.
  \end{itemize}
  The server only starts vacations when the system is empty.
  Unlike job arrivals, vacation start times are not a Poisson process.
  Each vacation has a \emph{type} $i \geq 1$ determining its length.
  Specifically, type~$i$ vacations have i.i.d. lengths drawn from distribution
  \begin{equation*}
    V_i = (\Xold[i]{\rJw} \mid \Xold[i]{\rJw} > 0).
  \end{equation*}
  The stationary probability that the the server is on a type~$i$ vacation
  is $\lambda\E{\Xold[i]{\rJw}}$.
  The exact process determining when a type~$i$ vacation starts is intractable,
  but to analyze queueing time in the VT~System,
  as we do in \fref{lem:waiting_time},
  it fortunately suffices to know just this stationary probability.
\end{namedobservation*}

\section{Response Time of SOAP Policies}
\label{sec:response_time}

Having spent the previous section understanding
the perspective of a tagged job in a system using a SOAP policy,
we are now ready to apply our insights
to analyze the response time of SOAP policies.
Specifically, we analyze~$T_{d, x}$,
the response time of a tagged job~$J$ with descriptor~$d$ and size~$x$.

When analyzing traditional policies like SRPT or PSJF (\fref{ex:soap_srpt}),
it often helps to think of $J$'s response time as the sum of
two independent random variables
\citep{srpt_analysis_schrage, book_harchol-balter}:
\begin{itemize}
\item
  \emph{waiting time}, written $\Twait_{d, x}$,
  the time from when $J$ arrives to when it first enters service; and
\item
  \emph{residence time}, written $\Tres_{d, x}$,
  the time from when $J$ first enters service to when it exits.
\end{itemize}
Unfortunately, the usual strategy for analyzing waiting and residence time
relies on $J$'s rank never increasing,
which holds for SRPT and PSJF but does not hold for a great many SOAP policies.
To overcome this obstacle,
we \emph{replace $J$'s rank with its worst future rank},
which never increases.
This rank substitution, fully justified in \fref{app:rank_substitution},
is made possible by the Pessimism Principle (\fref{sub:pessimism_principle}).
We call $J$ with its adjusted rank the \emph{rank-substituted} tagged job.

After replacing $J$'s rank with its worst future rank,
it remains to analyze its waiting and residence times.
Even though $J$'s worst future rank is monotonic,
other jobs' ranks may both increase and decrease with age,
making these analyses challenging.
\begin{itemize}
\item
  We can think of \emph{waiting time} as a transformed busy period:
  the initial transformed work is $J$'s total delay due to old jobs,
  and each arriving new job's transformed size is the amount it delays~$J$.
  The main challenge is finding the distribution of initial transformed work,
  which requires using the Vacation Transformation
  (\fref{sub:vacation_transformation}).
  We analyze waiting time in \fref{sub:waiting_time}.
\item
  We would also like to think of \emph{residence time}
  as a transformed busy period.
  The initial work is simply $J$'s size~$x$,
  but now the arriving new jobs present a challenge:
  because $J$'s worst future rank may decrease with time,
  not all new jobs have the same transformed size distribution.
  We analyze residence time in \fref{sub:residence_time}.
\end{itemize}

\subsection{Waiting Time}
\label{sub:waiting_time}

Hereafter, ``transform'' means Laplace-Stieltjes transform.
We write the transform of a random variable~$V$ as $\tl{V}(s)$
and the transform of a parametrized random variable~$V[r]$
as $\tl{V}[r](s)$.

Recall that the waiting time of the rank-substituted tagged job~$J$
is the amount of time between $J$'s arrival and when $J$ first enters service.
As mentioned previously, we can think of waiting time
as a transformed busy period:
the initial work is $J$'s total delay due to old jobs,
and each arriving new job's size is the amount it delays~$J$.
This type of busy period is formalized in the following definition.

\begin{definition}
  Let $r$ be a rank.
  The \emph{new $r$-work busy period}, written $\Bnew{r}$,
  is the length of a busy period
  in an M/G/1 system with arrival rate $\lambda$
  and job size $\Xnew{r}$.
  Its transform satisfies \citep{book_harchol-balter}
  \begin{equation*}
    \tlBnew{r}{s} = \tlXnew{r}{s + \lambda(1 - \tlBnew{r}{s})}.
  \end{equation*}
  More generally, the new $r$-work busy period
  \emph{started by work~$W$\esub}, written $\Bnew[W]{r}$,
  is the length of a busy period in the same M/G/1 system
  with a random initial amount of work~$W$.
  It has transform \citep{book_harchol-balter}
  \begin{equation}
    \label{eq:busy_work}
    \tlBnew[W]{r}{s} = \tl{W}(s + \lambda(1 - \tlBnew{r}{s})).
  \end{equation}
\end{definition}

Let $r = \rworst[d, x]{0}$ be $J$'s worst future rank upon arrival.
Because $J$ is not served until its residence time,
$r$ remains $J$'s worst future rank for the entirety of $J$'s waiting time.
This means $J$'s delay due to each new job that arrives during its waiting time
is $\Xnew{r}$,
so $J$'s waiting time is a new $r$-work busy period
started by initial work~$W$,
where $W$ is $J$'s total delay due to old jobs.

All that remains is to determine~$W$,
for which it is convenient to define two new notations.
First, let
\begin{align*}
  \rhonew{r} &= \lambda\E{\Xnew{r}} \\
  \rhoold[i]{r} &= \lambda\E{\Xold[i]{r}} \\
  \rhoold[\Sigma]{r} &= \sum_{i = 0}^\infty \rhoold[i]{r}
\end{align*}
be the ``loads contributed by'' new $r$-work, $i$-old $r$-work,
and all old $r$-work, respectively.
Second, let $\Yold[i]{r}$ be
the \emph{equilibrium distribution} \citep{book_harchol-balter},
or length-biased sample, of $\Xold[i]{r}$.
It has transform
\begin{equation*}
  \tlYold[i]{r}{s} = \frac{1 - \tlXold[i]{r}{s}}{s\E{\Xold[i]{r}}}.
\end{equation*}
Note that $\Yold[i]{r}$ is also the equilibrium distribution of
the length of a type~$i$ vacation in the VT~System,
$V_i = (\Xold[i]{r} \mid \Xold[i]{r} > 0)$,
because samples of length~$0$ are never encountered in a length-biased sample.

\begin{lemma}
  \label{lem:waiting_time}
  Under any SOAP policy,
  the Laplace-Stieltjes transform of waiting time
  for a rank-substituted tagged job with descriptor~$d$ and size~$x$ is
  \begin{equation*}
    \tlTwait_{d, x}(s)
    = \frac{%
        1 - \rhoold[\Sigma]{r}
        + \sum_{i = 1}^\infty \rhoold[i]{r}\tlYold[i]{r}{\sigma}}{%
        1 - \rhoold[0]{r}\tlYold[0]{r}{\sigma}},
  \end{equation*}
  where $r = \rworst[d, x]{0}$ and $\sigma = s + \lambda(1 - \tlBnew{r}{s})$.
\end{lemma}

\begin{proof}
  Call the rank-substituted tagged job~$J$.
  As previously mentioned,
  $\Twait_{d, x}$ is a new $r$-work busy period started by initial work~$W$,
  where $W$ is $J$'s total delay due to old jobs.
  This means $\tlTwait_{d, x}(s) = \tl{W}(\sigma)$ by~\fref{eq:busy_work},
  so it remains only to compute $\tl{W}(\sigma)$.

  The Vacation Transformation states that $W$ has the same distribution
  as the queueing time in the VT~System,
  which is a particular M/G/1/FCFS system with vacations.
  A decomposition result of \citet[Equation~(4)]{vacations_fuhrmann}
  states that the number of jobs in an M/G/1/FCFS system with vacations,
  such as the VT~System,
  is distributed as the sum of two independent random variables:
  \begin{itemize}
  \item
    the number~$N_Q$ of jobs in the queue
    of a vacation-free M/G/1/FCFS system; and
  \item
    the number~$N_V$ of jobs in the queue observed by an arriving job
    conditional on observing a non-busy server.
  \end{itemize}
  Recalling the Vacation Transformation
  and a standard result for the M/G/1 queue \citep{book_harchol-balter},
  we immediately obtain the probability generating function for~$N_Q$,
  \begin{equation}
    \label{eq:pgf_nq}
    \hat{N}_Q(z)
    = \frac{1 - \rhoold[0]{r}}{1 - \rhoold[0]{r}\tlYold[0]{r}{\lambda(1 - z)}}.
  \end{equation}

  In the VT~System, a job arriving to a non-busy server
  observes $N_V$ as one of the following.
  \begin{itemize}
  \item
    If the server is idle,
    which happens with probability $1 - \rhoold[\Sigma]{r}$,
    there are $0$ jobs in the queue.
  \item
    If the server is in the middle of a type~$i$ vacation,
    which happens with probability $\rhoold[i]{r}$,
    the number of jobs in the queue is
    the number of Poisson arrivals at rate~$\lambda$ during time~$\Yold[i]{r}$,
    which is the amount of time since the start of the type~$i$ vacation.
  \end{itemize}
  Accounting for the fact that
  we measure $N_V$ only when a job arrives to a non-busy server,
  which happens with probability $1 - \rhoold[0]{r}$,
  we obtain the probability generating function of $N_V$,
  \begin{equation}
    \label{eq:pgf_a}
    \hat{N}_V(z)
    = \frac{%
        1 - \rhoold[\Sigma]{r}
        + \sum_{i = 1}^\infty \rhoold[i]{r}\tlYold[i]{r}{\lambda(1-z)}}{%
        1 - \rhoold[0]{r}}.
  \end{equation}
  Multiplying \fref{eq:pgf_nq} and~\fref{eq:pgf_a},
  applying the distributional version of Little's Law\footnote{%
    Here we apply the distributional version of Little's Law to the VT~System,
    not the original SOAP system,
    which is crucial because the law applies only to FCFS systems.}
  \citep{distributional_little_keilson},
  and substituting $\sigma = \lambda(1 - z)$
  gives the transform of queueing time in the VT~System,
  which matches the desired transform $\tl{W}(\sigma)$.
\end{proof}

\subsection{Residence Time}
\label{sub:residence_time}

\begin{lemma}
  \label{lem:residence_time}
  Under any SOAP policy,
  the Laplace-Stieltjes transform of residence time
  for a rank-substituted tagged job with descriptor~$d$ and size~$x$ is
  \begin{equation*}
    \tlTres_{d, x}(s)
    = \exp\biggl(
        -\lambda \int_0^x (1 - \tlBnew{\rworst[d, x]{a}}{s}) \, da
      \biggr).
  \end{equation*}
\end{lemma}

\begin{proof}
  We view residence time as a \emph{sum of many small busy periods},
  each started by a small amount of work~$\delta$.
  We then take the $\delta \to 0$ limit,
  which exists thanks to conditions in \fref{app:rank_function}.
  This is very similar to the argument used to compute the transform of
  residence time under SRPT \citep{srpt_analysis_schrage}.
  Specifically, we divide $[0, x]$ into chunks of size~$\delta$
  and consider each small busy period started by the work needed
  to bring the job from age~$a$ to age $a + \delta$.
  In the $\delta \to 0$ limit,
  for the entirety of the small busy period starting at age~$a$,
  we can assume the job has rank $\rworst[d, x]{a}$.

  By the Pessimism Principle,
  the amount of work in the small busy period starting at age~$a$
  is the length of a new $\rworst[d, x]{a}$-work busy period
  started by work~$\delta$,
  which by~\fref{eq:busy_work} has transform
  \begin{equation*}
    \tlBnew[\delta]{\rworst[d, x]{a}}{s}
    = \exp(-\delta\lambda(1 - \tlBnew{\rworst[d, x]{a}}{s})).
  \end{equation*}
  The lengths of the small busy periods are independent
  because the arrival process is Poisson,
  so the residence time,
  which is the total amount of work in all such busy periods,
  has transform
  \begin{equation*}
    \prod_{i = 0}^{\mathclap{x/\delta}}
      \tlBnew[\delta]{\rworst[d, x]{\delta i}}{s}
    = \exp\biggl(
        -\delta\lambda \sum_{i = 0}^{\mathclap{x/\delta}}
          (1 - \tlBnew{\rworst[d, x]{\delta i}}{s})
      \biggr).
  \end{equation*}
  Taking the $\delta \to 0$ limit yields the desired expression.
\end{proof}

\subsection{Total Response Time}

\begin{theorem}[SOAP Transform of Response Time]
  \label{thm:response_time_transform}
  Under any SOAP policy,
  the Laplace-Stieltjes transform of response time
  of jobs with descriptor~$d$ and size~$x$ is
  \begin{equation*}
    \tl{T}_{d, x}(s) = \tlTwait_{d, x}(s)\tlTres_{d, x}(s),
  \end{equation*}
  where $\tlTres_{d, x}(s)$ and $\tlTwait_{d, x}(s)$ are as in
  Lemmas~\ref{lem:waiting_time} and~\ref{lem:residence_time}, respectively.
\end{theorem}

\begin{proof}
  Because the arrival process is Poisson,
  $\Twait_{d, x}$ and $\Tres_{d, x}$ are independent,
  so the result follows
  from Lemmas~\ref{lem:waiting_time} and~\ref{lem:residence_time}
\end{proof}

\begin{theorem}[SOAP Mean Response Time]
  \label{thm:response_time_mean}
  Under any SOAP policy,
  the mean response time of jobs with descriptor~$d$ and size~$x$ is
  \begin{equation*}
    \E{T_{d, x}}
    = \frac{%
          \lambda \sum_{i = 0}^\infty \E{(\Xold[i]{r})^2}}{%
          2(1 - \rhoold[0]{r})(1 - \rhonew{r})}
      + \int_0^x \frac{1}{1 - \rhonew{r(a)}} \, da,
  \end{equation*}
  where $r(a) = \rworst[d, x]{a}$ and $r = \rworst[d, x]{0}$.
\end{theorem}

\begin{proof}
  This follows from $\E{T_{d, x}} = -\tl{T}_{d, x}'(0)$
  after straightforward computation.
\end{proof}

We can use $T_{d, x}$ to analyze $T_d$,
the response time of jobs with descriptor~$d$ and any size,
and $T$, the overall response time of all jobs.
Specifically,
$\tl{T}_d(s) = \E{\tl{T}_{d, X_d}(s)}$,
where $X_d$ is the size distribution of jobs with descriptor~$d$,
and $\tl{T}(s) = \E{\tl{T}_D(s)}$,
where $D$ is the random descriptor assigned to a new job.
Analyzing response time of jobs with size~$x$ and any descriptor is similar,
but it requires computing $D_x$,
the distribution of descriptors of jobs with size~$x$.

\section{New Analyses for Specific Policies}
\label{sec:applications}

In this section, we analyze the response time
of several policies discussed in \fref{sub:examples_new}.
The main challenge to analyzing SOAP policies
is determining $\Xnew{r}$ and $\Xold[i]{r}$.
Throughout, we give expressions for these
focusing only on ranks~$r$ that are important for the final result.
Specifically,
for the possible descriptors~$d$ and sizes~$x$,
we only need to find
\begin{itemize}
\item
  $\Xnew{\rworst[d, x]{a}}$ for all ages~$a$ and
\item
  $\Xold[i]{\rworst[d, x]{0}}$.
\end{itemize}
For simplicity, we give final formulas only for mean response time.

\subsection{Discretized~FB}
\label{sub:discretized_fb}

Consider the discretized~FB policy (\fref{ex:soap_checkpoints}),
which can only preempt jobs when their age is at
specific checkpoints spaced $1$~time unit apart.
We represent this using the rank function
\begin{equation*}
  r(\dnil, a) = \angle{\floor{a} - a, a},
\end{equation*}
shown in \fref{fig:rank_checkpoints}.
As in the analysis of traditional FB,
it is convenient to work in terms of the \emph{capped size distribution}
and its associated load \citep[Section~30.3]{book_harchol-balter},
\begin{equation}
  \label{eq:capped}
  \begin{aligned}
    X_{\ol{x}} &= \min\{X, x\} \\
    \rho_{\ol{x}} &= \lambda\E{X_{\ol{x}}}.
  \end{aligned}
\end{equation}

In discretized~FB, a job of size~$x$ and age~$a$ has worst future rank
\begin{equation*}
  \rworst[\dnil, x]{a} =
  \begin{cases}
    \angle{0, \floor{x}} & \text{if } a \leq \floor{x} \\
    \angle{\floor{a} - a, a} & \text{if } a > \floor{x}.
  \end{cases}
\end{equation*}
The new $r$-work and $0$-old $r$-work relevant for obtaining response time
are therefore
\begin{align*}
  \Xnew{\angle{0, \floor{x}}} &= X_{\olfloor{x}} \\
  \Xnew{\angle{\floor{a} - a, a}} &= 0 \\
  \Xold[0]{\angle{0, \floor{x}}} &= X_{\olfloorplus{x}{1}{2.2mu}},
\end{align*}
where $\floor{x} < a < x$ in the second equation.
From \fref{thm:response_time_mean}, we conclude the following.

\begin{proposition}
  \label{prop:discretized_fb}
  Under discretized~FB (\fref{ex:soap_checkpoints}),
  the mean response time of jobs with size~$x$ is\footnote{%
    See \fref{app:rworst} for discussion of why
    we use $\ceil{x}$ instead of $\floor{x} + 1$.}
  \begin{equation*}
    \E{T_x}
    = \frac{%
        \lambda\mathbf{E}\bigl[X_{\olceil{x}}^2\bigr]}{%
        2\bigl(1 - \rho_{\olceil{x}}\bigr)\bigl(1 - \rho_{\olfloor{x}}\bigr)}
      + \frac{\floor{x}}{1 - \rho_{\olfloor{x}}} + x - \floor{x}.
  \end{equation*}
\end{proposition}

\subsection{Mixture of Known and Unknown Job Sizes}
\label{sub:humans_and_robots}

Consider the ``humans and robots'' system (\fref{ex:soap_humans_and_robots}).
This scenario has two features that were previously difficult to analyze.
\begin{itemize}
\item
  Some jobs have known size, namely robots,
  while others have unknown size, namely humans.
\item
  Some jobs are preemptible, namely robots,
  while others are nonpreemptible, namely humans.
\end{itemize}
Let
\begin{itemize}
\item
  $X_H$ and $X_R$ be the respective size distributions of humans and robots,
\item
  $p_H$ and $p_R = 1 - p_H$ be the respective
  probabilities that a given arrival is a human or a robot, and
\item
  $\lambda_H = \lambda p_H$ and $\lambda_R = \lambda p_R$ be the respective
  arrival rates of humans and robots.
\end{itemize}
Recall that humans all have descriptor $[H, ?]$,
indicating that their size is unknown,
and robots each have a descriptor of the form $[R, x]$,
indicating that their exact size is~$x$.
The rank function is
\begin{align*}
  r([H, ?], a) &= \angle{-a, x_H} \\
  r([R, x], a) &= \angle{0, x - a},
\end{align*}
where $x_H$ is a constant.

Both humans and robots have maximal primary rank, namely~$0$,
upon entering the system,
so the only nonzero new $r$-work and $i$-old $r$-work
occur for ranks of the form $r = \angle{0, x}$:
\begin{align*}
  \Xnew{\angle{0, x}} &=
  \begin{cases}
    X_H \1(x_H < x) & \text{with probability } p_H \\
    X_R \1(X_R < x) & \text{with probability } p_R
  \end{cases} \\
  \Xold[0]{\angle{0, x}} &=
  \begin{cases}
    X_H \1(x_H \leq x) & \text{with probability } p_H \\
    X_R \1(X_R \leq x) & \text{with probability } p_R
  \end{cases} \\
  \Xold[1]{\angle{0, x}} &=
  \begin{cases}
    X_H \1(x_H > x) & \text{with probability } p_H \\
    x \1(X_R > x) & \text{with probability } p_R,
  \end{cases}
\end{align*}
where $\1$ is the indicator function.
These follow from arguments similar to those given for SRPT
in Examples~\ref{ex:new_work_srpt} and~\ref{ex:old_work_srpt}.
Finally, $\Xold[i]{r} = 0$ for $i \geq 2$,
so \fref{thm:response_time_mean} yields the following.

\begin{proposition}
  In the humans and robots system (\fref{ex:soap_humans_and_robots}),
  the mean response time of humans is
  \begin{equation*}
    \E{T_H}
    = \frac{%
        \lambda_H\E{X_H^2} + \lambda_R\E{(X_R)_{\ol{x_H}}^2}}{%
        2(1 - \rho_H - \rho_{R \leq x_H})(1 - \rho_{R < x_H})}
      + \E{X_H},
  \end{equation*}
  and the mean response time of robots with size~$x$ is
  \iftoggle{widecol}{\begin{equation*}}{\begin{align*}}
    \E{T_{R, x}}
    \iftoggle{widecol}{}{&}= \frac{%
        \lambda_H\E{X_H^2} + \lambda_R\E{(X_R)_{\ol{x}}^2}}{%
        2(1 - \rho_H\1(x_H \leq x) - \rho_{R \leq x})(1 - \rho_{R < x_H})}
      \iftoggle{widecol}{}{\\ &\quad}
      + \int_0^x
          \frac{1}{1 - \rho_H\1(x_H \leq t) - \rho_{R < t}}
        \, dt,
  \iftoggle{widecol}{\end{equation*}}{\end{align*}}
  where $(X_R)_{\ol{x_H}}\esub$ and $(X_R)_{\ol{x}}\esub$
  are capped distributions as defined in~\fref{eq:capped} and
  \begin{align*}
    \rho_H &= \lambda_H \E{X_H} \\
    \rho_{R < x} &= \lambda_R \E{X_R \1(X_R < x)} \\
    \rho_{R \leq x} &= \lambda_R \E{X_R \1(X_R \leq x)}.
  \end{align*}
\end{proposition}

\subsection{The Gittins Index Policy}

As discussed in \fref{ex:soap_gittins},
the Gittins index policy can have a nonmonotonic rank function,
and thus only special cases of it have been analyzed in the past
\citep{multiclass_ayesta, rs_slowdown_hyytia}.
Theorems~\ref{thm:response_time_transform} and~\ref{thm:response_time_mean}
give us exactly the framework we need to analyze
the Gittins index policy
used with \emph{any} set of descriptors and job size distributions,
though $\Xnew{r}$ and $\Xold[i]{r}$ do not have a general closed form
and thus require details of the system to derive.
We start in this section by using the SOAP framework to analyze
the model considered by \citet{multiclass_ayesta},
which is relatively simple thanks to the fact that
the rank functions involved are monotonic.
In \fref{sub:gittins_serpt},
we move to a more difficult setting
in which the Gittins index policy has a nonmonotonic rank function.

We consider a system with two job classes $A$ and~$B$,
which serve as our descriptors,
with respective arrival rates
$\lambda_A = \lambda p_A$ and $\lambda_B = \lambda p_B$
and respective Pareto size distributions $X_A$ and~$X_B$.
Specifically,
\begin{equation*}
  \P{X_A > t} = \biggl(1 + \frac{t}{\beta_A}\biggr)^{-\alpha_A}
\end{equation*}
and symmetrically for~$B$,
where $\alpha_A, \alpha_B > 1$ and $\beta_A, \beta_B > 0$
are parameters of the distributions.
The rank function is \citep{m/g/1_gittins_aalto, multiclass_ayesta}
\begin{equation*}
  r(A, a) = \frac{1}{G(A, a)} = \frac{\beta_A + a}{\alpha_A}
\end{equation*}
and symmetrically for~$B$.
The rank is strictly increasing in~$a$,
so $\rworst[d, x]{a} = r(d, x)$ for all ages $a < x$.
Strictly increasing rank also means jobs are never recycled,
so $\Xold[i]{r} = 0$ for all $i \geq 1$.

It remains only to compute $\Xnew{r}$ and $\Xold[0]{r}$.
As in \fref{def:new_work}, let
\begin{equation*}
  c_A[r] = \max\{\alpha_A r - \beta_A, 0\}
\end{equation*}
be the age at which a class~$A$ job first surpasses rank~$r$,
and symmetrically for~$B$.
Note that $c_A[r(A, x)] = x$.
Because the rank function is strictly increasing,
\begin{equation*}
  \Xnew{r} = \Xold[0]{r} =
  \begin{cases}
    (X_A)_{\ol{c_A[r]}} & \text{with probability } p_A \\
    (X_B)_{\ol{c_B[r]}} & \text{with probability } p_B,
  \end{cases}
\end{equation*}
where $(X_A)_{\ol{c_A[r]}}$ and $(X_B)_{\ol{c_B[r]}}$
are capped distributions as defined in~\fref{eq:capped}.
By \fref{thm:response_time_mean},
\begin{equation*}
  \E{T_{A,x}}
  = \frac{%
      \lambda_A\E{(X_A)_{\ol{x}}^2}
      + \lambda_B\E{(X_B)_{\ol{y}}^2}}{%
      2(1 - \rho_{\ol{x}, \ol{y}})^2}
    + \frac{x}{1 - \rho_{\ol{x}, \ol{y}}},
\end{equation*}
where
\begin{align*}
  y &= c_B[r(A, x)] \\
  \rho_{\ol{x}, \ol{y}}
  &= \lambda_A\E{(X_A)_{\ol{x}}}
    + \lambda_B\E{(X_B)_{\ol{y}}}.
\end{align*}
Of course, $\E{T_{B, x}}$ is symmetrical.
It is simple to verify that, modulo notation,
this matches the results of \citet{multiclass_ayesta}.

\subsection{Case Study: SERPT vs. the Gittins Index}
\label{sub:gittins_serpt}

We now analyze both SERPT and the Gittins index policy
on the size distribution introduced in \fref{ex:soap_serpt}.
Both policies have nonmonotonic rank functions in this case,
so we need the full power of \fref{thm:response_time_mean}
to compute mean response times.
All jobs have descriptor~$\dnil$
and the same two-point size distribution.
We write the size distribution as $X = \2\{2, 14\}$,
meaning jobs have size~$2$ with probability~$1/2$ and size~$14$ otherwise.

We begin by analyzing SERPT.
We have computed $\Xnew{r}$ for SERPT with this size distribution
in \fref{ex:new_work_serpt},
\begin{equation*}
  \Xnew{r} =
  \begin{cases}
    0 & r \leq 8 \\
    2 & r > 8
  \end{cases}
\end{equation*}
We only need to compute $i$-old $r$-work
for ranks $\rworst[\dnil, 2]{0} = 8$
and $\rworst[\dnil, 14]{0} = 12$.
From the rank function plot in \fref{fig:old_work_serpt}, we see
\begin{align*}
  \Xold[0]{8} &= 2 \\
  \Xold[0]{12} &= \2\{2, 14\} \\
  \Xold[1]{8} &= \2\{0, 8\} \\
  \Xold[1]{12} &= 0.
\end{align*}
The most subtle of these is $\Xold[1]{8}$:
the total time an old job~$I$ spends as recycled with respect to rank~$8$
is either~$0$, if $I$ has size~$2$,
or~$8$, if $I$ has size~$14$.
Finally, $\Xold[i]{r} = 0$ for $i \geq 2$.
Applying \fref{thm:response_time_mean} yields size-specific mean response times
\begin{align*}
  \E{T_2^{\mathrm{SERPT}}}
  &= \frac{18\lambda}{1 - 2\lambda} + 2 \\
  \E{T_{14}^{\mathrm{SERPT}}}
  &= \frac{50\lambda}{(1 - 8\lambda)(1 - 2\lambda)}
    + \frac{6}{1 - 2\lambda} + 8.
\end{align*}

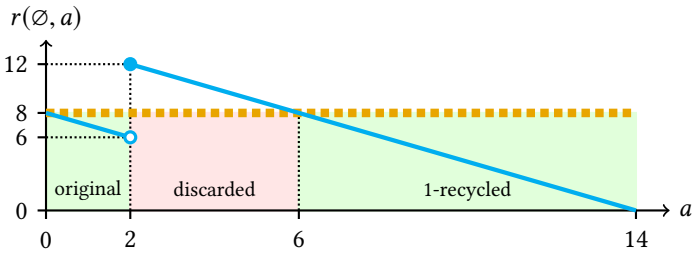
\begin{figure}
  \centering
  \begin{tikzpicture}[rank plot]
  \original{(0, 8)}{(2, 8)}
  \discarded{(2, 8)}{(6, 8)}
  \recycled[1]{(6, 8)}{(14, 8)}

  \xguide{2}{12}
  \xguide{6}{8}
  \xguide{14}{0}
  \yguide{2}{6}
  \yguide{0}{8}
  \yguide{2}{12}

  \axes{14.8}{14}{$a$}{$r(\emptyset, a)$}

  \draw[cutoff] (0, 8) -- (14, 8);

  \draw[primary] (0, 8) -- (2, 6);
  \draw[primary] (2, 12) -- (14, 0);

  \jump{2}{6}{12}
\end{tikzpicture}

  \captiondetail
  The rank function for SERPT where jobs have
  two-point size distribution $\2\{2,14\}$,
  meaning size~$2$ with probability $1/2$ and size~$14$ otherwise.
  The same rank function appears in \fref{fig:rank_serpt}.
  The $0$-old (original) and $1$-old ($1$-recycled) intervals
  are highlighted in green.
  A job is original with respect to rank~$8$ until age~$2$,
  when its rank jumps up if it does not complete.
  Upon reaching age~$6$, the job has remaining size~$8$,
  so it becomes recycled.

  \caption{Original and Recycled Work in SERPT}
  \label{fig:old_work_serpt}
\end{figure}

The Gittins index policy for the same system
has rank function $r(\dnil, a) = 1/G(\dnil, a)$,
where, by the definition in \fref{ex:soap_gittins},
\begin{equation*}
  \frac{1}{G(\dnil, a)} =
  \begin{cases}
    4 - 2a & \text{if } a < 2 \\
    14 - a & \text{if } a \geq 2.
  \end{cases}
\end{equation*}
This rank function is illustrated in \fref{fig:rank_gittins}.
Broadly speaking, the Gittins index policy
places higher priority on jobs of age $a < 2$ than SERPT does.
Like that of SERPT, the rank function is piecewise linear with negative slopes,
so we omit the very similar analysis and simply state
the size-specific mean response times:
\begin{align*}
  \E{T_2^{\mathrm{Gittins}}}
  &= \frac{6\lambda}{1 - 2\lambda} + 2 \\
  \E{T_{14}^{\mathrm{Gittins}}}
  &= \frac{50\lambda}{(1 - 8\lambda)(1 - 2\lambda)}
    + \frac{10}{1 - 2\lambda} + 4.
\end{align*}
As expected due to its prioritization of jobs of age $a < 2$,
the Gittins index policy has shorter mean response time for jobs of size~$2$
but longer mean response time for jobs of size~$14$.
The Gittins index policy is known to minimize overall mean response time,
and it performs as promised:
\begin{align*}
  \E{T_2^{\mathrm{SERPT}}} - \E{T_2^{\mathrm{Gittins}}}
  &= \frac{12\lambda}{1 - 2\lambda} \\
  \E{T_{14}^{\mathrm{SERPT}}} - \E{T_{14}^{\mathrm{Gittins}}}
  &= \frac{-8\lambda}{1 - 2\lambda}.
\end{align*}
Because the two job sizes are equally likely,
the Gittins index policy has lower overall mean response time than SERPT.

\section{Conclusion}

We introduce \emph{SOAP policies},
a very broad class of scheduling policies for the M/G/1 queue.
The characteristic feature of a SOAP policy is its \emph{rank function},
which maps each possible state a job could be in
to a \emph{rank}, meaning priority level.
The SOAP class includes many policies old and new.
While the mean response times
of some relatively simple SOAP policies have been analyzed previously,
the vast majority of SOAP policies,
in particular those with \emph{nonmonotonic} rank functions,
have resisted analysis.
Using two key technical insights,
the \emph{Pessimism Principle} and the \emph{Vacation Transformation},
we overcome the obstacles presented by nonmonotonic rank functions
to present a \emph{universal response time analysis}
that applies to any SOAP policy.

Our universal analysis applies to some notable policies.
Among these is the \emph{Gittins index policy},
which has long been known to minimize mean response time
in settings where exact job sizes are not known.
While prior work \citep{multiclass_ayesta, rs_slowdown_hyytia}
was restricted to the case of known job sizes
or distributions with the decreasing hazard rate property,
our analysis can handle the Gittins index policy
with \emph{arbitrary size distributions},
which was previously intractable.
Our universal analysis also applies to several
\emph{practically motivated systems},
such as those in which jobs are only preemptible at certain checkpoints
or only some jobs' exact sizes are known.
More broadly, we are optimistic that techniques similar to
our Pessimism Principle and Vacation Transformation
could help analyze the response times of scheduling policies
in more complex M/G/1 settings,
such as systems with setup times or server vacations.

\begin{acks}
  We thank Peter van de Ven and the anonymous referees
  for their helpful comments.
  Ziv Scully was supported by an
  \grantsponsor{arcs}{ARCS Foundation}{https://www.arcsfoundation.org}
  scholarship and the
  \grantsponsor{nsf}{National Science Foundation}{https://www.nsf.gov}
  Graduate Research Fellowship Program under
  Grant No.~\grantnum{nsf}{DGE-1745016}.
  Mor Harchol-Balter was supported by
  NSF-\grantnum{nsf}{XPS-1629444},
  NSF-\grantnum{nsf}{CMMI-1538204},
  NSF-\grantnum{nsf}{CMMI-1334194},
  and a Faculty Award from
  \grantsponsor{g}{Google}{https://research.google.com/research-outreach.html}.
\end{acks}

\bibliographystyle{ACM-Reference-Format}
\bibliography{refs}

\appendix

\section{Extension to LCFS Tiebreaking}
\label{app:lcfs}

SOAP policies that use LCFS tiebreaking admit almost exactly the same analysis
as those that use FCFS tiebreaking.
As explained in detail below,
the entire analysis is unchanged except for
\emph{reversing the strictness of rank comparisons}
in Definitions~\ref{def:new_work} and~\ref{def:old_work},
meaning swapping $\prec$ and $\preceq$ and swapping $\succ$ and $\succeq$.

Throughout Sections~\ref{sub:pessimism_principle}
and~\ref{sub:vacation_transformation},
which follow a tagged job~$J$ through the system,
we distinguish between \emph{new} jobs, which arrive after~$J$,
and \emph{old} jobs, which arrive before~$J$.
When there are multiple jobs of minimal rank,
FCFS tiebreaking prioritizes old jobs, then $J$, and then new jobs.
This prioritization affects the \emph{strictness of rank comparisons}
when defining new $r$-work and old $r$-work.
For instance, \fref{def:new_work} defines
\begin{equation*}
  c_d[r] = \inf\{a \geq 0 \mid r(d, a) \succeq r\},
\end{equation*}
whereas \fref{def:old_work} defines
\begin{equation*}
  c_{0, d}[r] = \inf\{a \geq 0 \mid r(d, a) \succ r\},
\end{equation*}
which is the same but with $\succ$ in place of~$\succeq$.
When $J$'s worst future rank is~$r$,
the above values each represent a ``cutoff age''
before which a job of descriptor~$d$ outranks~$J$.
Under FCFS tiebreaking,
a new job~$K$ outranks $J$ until $K$'s rank is \emph{at least}~$r$,
whereas an old job~$I$ outranks $J$
until $I$'s rank \emph{strictly exceeds}~$r$.
Under LCFS tiebreaking, this situation is reversed,
which manifests as reversing the strictness of rank comparisons.

\section{Rank Function Details}
\label{app:rank_function}

In order to ensure that a SOAP policy is well-defined,
its rank function~$r$ must satisfy the following conditions.
\begin{itemize}
\item
  With respect to descriptor,
  $r$ must be \emph{piecewise continuous}
  to ensure that certain expectations are well-defined.
\item
  With respect to age,
  $r$ must be \emph{piecewise monotonic} and \emph{piecewise differentiable}
  to determine when and how to share the processor between multiple jobs.
  Any compact region of $\R_{\geq 0}$ must contain
  only finitely many boundary points between pieces.
  Furthermore, upwards jump discontinuities
  must be continuous from the right\footnote{%
    That is, if a job jumps from low rank to high rank at age~$a$,
    its rank exactly at age~$a$ should be the high rank.}.
\end{itemize}
These conditions allow us to define a SOAP policy
as the limit of discrete-time priority policies,
with the limit taken as the discretization increment approaches~$0$.
When $\mc{R} = \R^2$ ordered lexicographically,
the limiting policy is~\fref{alg:soap},
which has a clear generalization to $\mc{R} = \R^n$.
In \fref{alg:soap},
we say a job is ``in state $(d, a)$'' to mean
it has descriptor~$d$ and age~$a$.

\begin{algorithm}
  \caption{SOAP Policy in Continuous Time}
  \label{alg:soap}
  \justifying\noindent
  Let $\mc{J}$ be the set of jobs in states $(d, a)$ of minimal $r_1(d, a)$.
  \begin{itemize}
  \item
    Within~$\mc{J}$, consider jobs such that
    $r$ is strictly decreasing in age.
    If there are any, schedule the job of minimal~$r_2$,
    using FCFS tiebreaking if there are multiple such jobs.
  \item
    Otherwise, within~$\mc{J}$, consider jobs such that
    $r_1$ is constant and $r_2$ is strictly increasing in age.
    If there are any, share the processor between all such jobs,
    giving a job in state $(d, a)$ share proportional to
    $1 / \partial_ar_2(d, a)$.
  \item
    Otherwise, $\mc{J}$ must only contain jobs such that
    $r_1$ is strictly increasing.
    Share the processor between jobs in~$\mc{J}$,
    giving a job in state $(d, a)$ share proportional to
    $1 / \partial_ar_1(d, a)$.
  \end{itemize}
\end{algorithm}

\section{Worst Future Rank Details}
\label{app:rworst}

For simplicity of exposition,
throughout \fref{sec:key_ideas},
we assumed that a job's worst future rank $\rworst[d, x]{a}$
was actually attained by that job in the future.
However, there are two cases where the supremum in \fref{def:rworst}
is \emph{not} be attained by some age~$a$:
when there is a jump discontinuity
or when the maximum is at the open boundary $a = x$.
For example, in \fref{sub:discretized_fb},
if a job has integer size~$x$,
then the job never attains rank $\angle{0, x}$.

There are multiple ways to remedy the situation.
The most intuitive is to say that $\rworst[d, x]{a}$
is not a rank but a \emph{rank bound}.
The set of rank bounds is $\mc{R} \times \{-1, 0\}$
ordered lexicographically.
An ordinary rank~$r$ corresponds to the pair $(r, 0)$
representing the \emph{closed} upper bound $r' \preceq r$
over other ranks~$r'$,
but the pair $(r, -1)$ is ``just below'' $(r, 0)$,
representing the \emph{open} upper bound $r' \prec r$.

The corrections to definitions are as follows.
In \fref{def:rworst},
we define the worst future rank to be the rank bound
\begin{equation*}
  \rworst[d, x]{a}
  = (\sup_{\mathclap{a \leq b < x}} r(d, b),
    -\1(\text{the supremum is not attained}))
\end{equation*}
instead of just a rank.
In Definitions~\ref{def:new_work} and~\ref{def:old_work},
instead of defining $r$-work for a rank~$r$,
we define $(r, q)$-work for rank bounds $(r, q)$.
When we compare a rank $r'$ against a rank bound $(r, q)$,
we compare rank bound $(r', 0)$ against $(r, q)$.
Concretely, in \fref{def:new_work}, we define
\begin{equation*}
  c_d[(r, q)] = \inf\{a \geq 0 \mid (r(d, a), 0) \succeq (r, q)\},
\end{equation*}
and similarly for $b_{i, d}[(r, q)]$ and $c_{i, d}[(r, q)]$
in \fref{def:old_work}.

For example,
consider the analysis of discretized~FB in \fref{sub:discretized_fb}.
When a job has integer size~$x$,
the supremum in $\rworst[\dnil, x]{a}$ is attained only in the $b \to x$ limit.
Thus, $\Xold[0]{\rworst[\emptyset,x]{0}} = X_{\ol{x}}$ when $x$ is an integer,
not $X_{\ol{x + 1}}$ as would follow from the uncorrected \fref{def:old_work}.
To correct for this,
\fref{prop:discretized_fb} uses $\ceil{x}$ instead of $\floor{x} + 1$.

The above discussion assumes FCFS tiebreaking.
As discussed in \fref{app:lcfs},
the strictness of rank comparisons in
Definitions~\ref{def:new_work} and~\ref{def:old_work}
is reversed under LCFS tiebreaking,
but the same changes described above apply without issue.

\section{The Rank-Substituted Tagged Job}
\label{app:rank_substitution}

In this section, we show that for the purposes of analyzing $T_{d, x}$,
we can use a \emph{rank-substituted} tagged job.

\begin{namedobservation*}{Worst Future Rank Substitution}
  Consider an arbitrary arrival sequence
  that includes the arrival of a tagged job~$J$
  with descriptor~$d$ and size~$x$.
  The response time of~$J$ is unaffected if we
  \emph{schedule $J$ as if its rank were $\rworst[d, x]{a}$
    instead of $r(d, a)\esub$}
  at every age~$a$,
  without otherwise changing the arrival sequence or scheduling policy.
  We call this process \emph{rank substitution}.
\end{namedobservation*}

Worst Future Rank Substitution is
a direct consequence of the Pessimism Principle
(\fref{sub:pessimism_principle}).
To see why it holds,
consider two systems experiencing identical job arrivals,
including tagged job~$J$.
\begin{itemize}
\item
  \emph{System~A} is unmodified,
  so $J$ is scheduled according to its current rank.
\item
  \emph{System~B} uses rank substitution,
  so $J$ is scheduled according to its worst future rank.
\end{itemize}
We say the two systems \emph{synchronize} at time~$t$
if they contain the same jobs at the same ages at~$t$.
The systems clearly synchronize at $J$'s arrival time.
We show below that the systems also synchronize at many other points in time,
one of which is $J$'s exit time,
so $J$'s response time is the same in each system.

The Pessimism Principle states that
all of $J$'s delay due to another job~$L$
occurs \emph{before $J$ is served while at its worst future rank}.
This suggests we should focus on $J$'s \emph{worst future age},
which when $J$ has age~$a$ is
\begin{equation*}
  \aworst[d, x]{a} = \inf\{b \geq a \mid r(d, b) = \rworst[d, x]{b}\},
\end{equation*}
namely the earliest age at which $J$ attains its worst future rank\footnote{%
  As discussed in \fref{app:rworst},
  a job's worst future rank is sometimes only attained in a limit
  due to the job's completion or a jump in the rank function.
  Accounting for this changes only minor details in the following discussion.}.
See \fref{fig:aworst} for an illustration.
Note that $a = \aworst[d, x]{a}$ if and only if $r(d, a) = \rworst[d, x]{a}$.

\begin{figure}
  \centering
  \begin{tikzpicture}[rank plot]
  \begin{scope}[shift={(0, -45)}, y=10]
    \xguide[$v_1$]{5.5}{45*0.4 + 12*0.4}
    \xguide[$w_1$]{6.25}{45*0.4 + 8*0.4}
    \xguide[$v_2$]{8}{45*0.4 + 8*0.4}
    \xguide[$w_2$]{8.75}{45*0.4 + 5*0.4}
    \xguide[$x$]{10.25}{45*0.4 + 5*0.4}

    \yguide[$v_1$]{0}{5.5}
    \yguide[$w_1$]{6.25}{6.25}
    \yguide[$v_2$]{6.25}{8}
    \yguide[$w_2$]{8.75}{8.75}
    \yguide[$x$]{8.75}{10.25}

    \axes{14.8}{14}{$a$}{$\aworst[d, x]{a}$}

    \draw[aworst] (0, 5.5) -- (5.5, 5.5) -- (6.25, 6.25);
    \draw[aworst] (6.25, 8) -- (8, 8) -- (8.75, 8.75);
    \draw[aworst] (8.75, 10.25) -- (10.25, 10.25);

    \jump[aworst]{6.25}{8}{6.25}
    \jump[aworst]{8.75}{10.25}{8.75}
    \point[aworst, fill=white]{(10.25, 10.25)}
  \end{scope}

  \begin{scope}
    \axes{14.8}{14}{$a$}{$r(d, a)$}

    \Jsnake

    \draw[rworst]
    (0, 12) -- (5.5, 12)
    -- \snakefirst{(5.5, 12)}{(7, 4)}
    -- \snakefirst{(8, 8)}{(9.5, 2)}
    -- (10.25, 5);
    \point[rworst point]{(10.25, 5)}
    \node[above] at (3.5, 12) {$\rworst[d, x]{a}$};
  \end{scope}
\end{tikzpicture}

  \captiondetail
  The relationship between rank $r(d, a)$ (solid cyan),
  worst future rank $\rworst[d, x]{a}$ (dashed magenta),
  and worst future age $\aworst[d, x]{a}$ (solid green)
  for a job with descriptor~$d$ and size~$x$.
  Age and worst future age coincide for ages in the intervals
  $[v_1, w_1]$ and $[v_2, w_2]$.

  \caption{Illustration of Worst Future Age}
  \label{fig:aworst}
\end{figure}

As $J$ ages, its worst future age alternates between being
a constant future age and its current age.
Let $[v_i, w_i]$ be the $i$th interval of ages~$a$ such that
$a = \aworst[d, x]{a}$.
It is convenient to set $w_0 = 0$ and $v_{n + 1} = x$,
where $n$ is the number of $[v_i, w_i]$ intervals.
We will show that System~A and System~B synchronize when,
for some~$i$, either
\begin{itemize}
\item
  $J$ has age~$v_i$ and is in service or
\item
  $J$'s age is in $(v_i, w_i]$.
\end{itemize}
It is clear that if the systems synchronize
when $J$ is served at age~$v_i$,
then the systems remain synchronized until $J$ reaches age~$w_i$,
because $J$'s rank in the two systems is identical until $J$ reaches age~$w_i$.
Thus, it suffices to show that
if the systems synchronize when $J$ has age~$w_{i - 1}$,
then the systems synchronize when $J$ is served at age~$v_i$.

Consider how the two systems change during the interval between
their synchronization when $J$ has age~$w_{i - 1}$.
Let $t_{\mathrm{A}}$ be the moment when $J$ is served at age~$v_i$ in System~A,
and symmetrically for $t_{\mathrm{B}}$ in System~B.
\begin{itemize}
\item
  In System~A, by the Pessimism Principle,
  each other job~$L$ in the system during the interval is served until
  it either completes or surpasses $J$'s worst future rank $r(d, v_i)$,
  after which $L$ is never served again.
  This is because $\aworst[d, x]{a} = v_i$ for all $a \in (w_{i - 1}, v_i]$.
\item
  In System~B, by the Pessimism Principle,
  each other job~$L$ is served for the same amount of time as in System~A,
  because rank substitution does not change $J$'s worst future rank.
\end{itemize}
Thus, if System~A experiences the same arrivals before $t_{\mathrm{A}}$
that System~B experiences before~$t_{\mathrm{B}}$,
then the systems synchronize at $t_{\mathrm{A}} = t_{\mathrm{B}}$, as desired.
Suppose for contradiction that one system, say System~A,
experiences an extra arrival.
Because the arrival sequence is the same for the two systems,
this occurs only if $t_{\mathrm{A}} > t_{\mathrm{B}}$.
But by work conservation and the observations above, at~$t_{\mathrm{B}}$,
System~A must serve $J$ at age~$v_i$,
so $t_{\mathrm{A}} = t_{\mathrm{B}}$ after all,
contradicting $t_{\mathrm{A}} > t_{\mathrm{B}}$.

\end{document}